\documentclass[12pt,a4paper]{article}
\usepackage{amsmath}
\usepackage{amsfonts}
\usepackage{amsthm}
\usepackage{graphicx}
\usepackage{mathtools}
\usepackage{color}
\usepackage[ruled,linesnumbered]{algorithm2e}
\usepackage[colorlinks,citecolor=blue,urlcolor=blue]{hyperref}
\usepackage{bm}

\usepackage[T1]{fontenc}

\newcommand{\XX}{{\bm X }}
\newcommand{\YY}{{\bm Y }}

\newcommand{\xx}{{\bm  x}}
\newcommand{\yy}{{\bm y }}

\newcommand{\VV}{{\bm V }}
\newcommand{\WW}{{\bm W }}

\newcommand{\uu}{{\bm u }}
\newcommand{\vv}{{\bm v }}
\newcommand{\ww}{{\bm w }}

\newcommand{\UU}{{\bm  U}}

\def \Db{{\mathbb D}}
\def \Eb{{\mathbb E}}

\def \Pb{{\mathbb P}}
\def \Qb{{\mathbb Q}}
\def \Rb{{\mathbb R}}

\def \Ec{{\mathcal E}}

\def \Fc{{\mathcal F}}
\def \Cc{{\mathcal C}}
\def \Sc{{\mathcal S}}
\def \Hc{{\mathcal H}}

\def \Nc{{\mathcal N}}

\def \Ic{{\mathcal I}}
\def \Jc{{\mathcal J}}

\def \Uc{{\mathcal U}}
\def \Oc{{\mathcal O}}

\def \Cf{{\mathfrak C}}

\newcommand{\bqa}{\begin{eqnarray*}}
\newcommand{\eqa}{\end{eqnarray*}}
\newcommand{\bqan}{\begin{eqnarray}}
\newcommand{\eqan}{\end{eqnarray}}
\newcommand{\bqt}{\begin{quote}}
\newcommand{\eqt}{\end{quote}}
\newcommand{\bt}{\begin{tabbing}}
\newcommand{\et}{\end{tabbing}}
\newcommand{\bit}{\begin{itemize}}
\newcommand{\eit}{\end{itemize}}
\newcommand{\ben}{\begin{enumerate}}
\newcommand{\een}{\end{enumerate}}
\newcommand{\beq}{\begin{equation}}
\newcommand{\eeq}{\end{equation}}

\newtheorem{Theorem}{Theorem}

\newtheorem{Proposition}{Proposition}

\newtheorem{Corollary}{Corollary}
\newtheorem{Remark}{Remark}
\newtheorem{Example}{{\it Example}}

\newtheorem{Condition}{Condition}

\newcommand{\eps}{\varepsilon}

\def\mds{\medskip}
\def\1{{\mathbf 1}}
\def\0{{\mathbf 0}}

\title{Estimation of copulas via Maximum Mean Discrepancy}

\author{Pierre Alquier~\footnote{RIKEN AIP, Nihonbashi 1-chome Mitsui Building (15th floor), 1-4-1 Nihonbashi, Chuo-ku, Tokyo, 103-0027, Japan. Email: \texttt{pierrealain.alquier@riken.jp}}, Badr-Eddine Ch\'erief-Abdellatif~\footnote{Department of Statistics, University of Oxford. 24-29 St Giles'
Oxford OX1 3LB, United Kingdom. Email: \texttt{badr-eddine.cherief-abdellatif@stats.ox.ac.uk}}, Alexis Derumigny~\footnote{Delft University of Technology, Department of Applied Mathematics. Mekelweg 4, 2628 CD Delft. The Netherlands. Email: \texttt{a.f.f.derumigny@tudelft.nl}},\\ and Jean-David Fermanian\footnote{CREST, ENSAE, Institut Polytechnique de Paris, 5 Avenue Le Chatelier, 91120 Palaiseau, France. Email: \texttt{jean-david.fermanian@ensae.fr}}}

\date{\today}

\begin{document}

\maketitle

\begin{abstract}
This paper deals with robust inference for parametric copula models. Estimation using Canonical Maximum Likelihood might be unstable, especially in the presence of outliers. We propose to use a procedure based on the Maximum Mean Discrepancy (MMD) principle. We derive non-asymptotic oracle inequalities, consistency and asymptotic normality of this new estimator. In particular, the oracle inequality holds without any assumption on the copula family, and can be applied in the presence of outliers or under misspecification. Moreover, in our MMD framework, the statistical inference of copula models for which there exists no density with respect to the Lebesgue measure on $[0,1]^d$, as the Marshall-Olkin copula, becomes feasible. A simulation study shows the robustness of our new procedures, especially compared to pseudo-maximum likelihood estimation. An R package implementing the MMD estimator for copula models is available.
\end{abstract}

\section{Introduction}

\subsection{Context}

Since the seminal work of Sklar~\cite{sklar1959}, it is well known that every $d$-dimensional distribution $F$ can be decomposed as
$F(\xx) = C \big(F_1(x_1), \dots, F_d(x_d) \big)$, for all $\xx = (x_1, \dots, x_d) \in \Rb^d$.
Here, $F_1, \dots, F_d$ are the marginal distributions of $F$ and $C$ is a distribution on the unit cube $[0,1]^d$ with uniform margins, called a copula. This allows any user to split the complex problem of estimating a multivariate distribution into two simpler problems which are the estimation of the margins on one side, and of the copula on the other side.
Copulas have become increasingly useful to model multivariate distributions in a wide variety of applications : finance, insurance, hydrology, engineering and so on.
We refer to \cite{Nelsen2007,HofertKoja} for a general introduction and background on copula models.

\mds

Often, a copula of interest $C$ belongs to a parametric family $\Cc=\{C_\theta, \theta \in \Theta\subset \Rb^p\}$ and one is interested in the estimation of
the ``true'' value of the parameter~$\theta$. Typically, the goal is to evaluate the underlying copula only, without trying to specify the marginal distributions.
In such a case, the most popular method for estimating parametric copula models is by Canonical Maximum Likelihood or CML, shorter (\cite{Genest1995,ShihLouis}). This is a semi-parametric analog of Maximum Likelihood Estimation for copula models for which the margins are left unspecified and replaced by nonparametric counterparts.
The method of moments is also a popular estimation technique, most often when $p=1$,
and is usually done by inversion of Kendall's tau or Spearman's rho. The latter estimators have been implemented in the R package VineCopula~\cite{VineCopulaR} and attain the usual $\sqrt{n}$ rate of convergence as if the margins were known: see \cite{tsukahara2005} for the asymptotic theory.

\mds

Nevertheless, all the aforementioned estimation approaches suffer from drawbacks. In particular, they are not robust statistically speaking. More specifically, assume that the true copula is slightly perturbed in the sense that $C = (1 - \varepsilon) C_{\theta_0} + \varepsilon \tilde C$ for a small $\varepsilon > 0$ and a copula $\tilde C \neq C_{\theta_0}$. In general, there is no guarantee that the estimators obtained by CML or by the method of moments should be close to $\theta_0$ when $\varepsilon \neq 0$, since this problem still occurs in the case of most usual M-estimators generally speaking.

\mds

In the literature, there are very few attempts to build robust estimation methods for semi-parametric copula models that would be ``omnibus'' (i.e. not dependent on some particular choices of models).
Using Mahalanobis distances computed using robust estimates of covariance and location, Mendes et al.~\cite{MendesNelsen} identified some points which seem not to follow the assumed dependence structure. Then, some copula parameters are obtained through the minimization of weighted goodness of fit statistics.
In the semiparametric copula-based multivariate dynamic (SCOMDY) framework (\cite{ChenFan2006}), Kim and Lee~\cite{KimLee} built a minimum density power divergence estimator which shows some resistance to some types of outliers.
Deneke and M\"uller\cite{Denecke} proposed a parametric robust estimation method based on likelihood depth (\cite{Rousseeuw}).
Recently, Goegebeur et al.~\cite{Goegebeur} have considered robust and nonparametric estimation of the coefficient of tail dependence in presence of random covariates, that may be a way of estimating copulas for some particular models.
Therefore, even if many estimators have been proposed for Huber contaminated models in general parametric cases, this has not been the case for semiparametric copula models yet.
This paper is an attempt to fill this gap.

\mds

To this end, we need to consider a relevant distance between distributions.
The Maximum Mean Discrepancy (MMD) between two arbitrary probability distributions $\Pb$ and $\Qb$ is defined as
$$  \Db(\Pb,\Qb)= \sup_{f\in \Fc} \bigg| \int f\, d\Pb - \int f\, d\Qb \bigg|,$$
where $\Fc$ is the unit ball in a universal reproducing kernel Hilbert space (RKHS) $\Hc$ defined on a compact metric space, with an associated kernel $K$ and a norm $\|\cdot\|_{\Hc}$.
It can be proved that $\Db(\Pb,\Qb)$ is the distance between the kernel mean embeddings of the two underlying probabilities, i.e.
$  \Db(\Pb,\Qb)=\| \mu_\Pb - \mu_\Qb\|_{\Hc}$; see Muandet et al.~\cite{Muandet2017}, Section 3.5, that provides a state-of-the-art introduction to the theory of RKHS and MMD.
When the kernel $K$ is characteristic  (i.e. when the map $\Pb \mapsto \mu_\Pb$ is injective), MMD becomes a distance
between the two probabilities $\Pb$ and $\Qb$. Such a distance can be easily empirically estimated and has been used many times in different areas of statistics and machine learning; see, e.g.,~\cite{Danafar2013,Gretton2012} for the two-sample test problem.

\mds

As a tool for parametric estimation, MMD has been studied as a general method for inference only recently~\cite{gerber2020,Briol2019,Cherief2019,Cherief2020}, even though it was implicitly used in specific examples in machine learning (\cite{Dziugaite2015}). In the latter papers, it appeared that MMD criteria lead to consistent estimators that are
robust to model misspecification, for most models and without any assumption on the actual distribution of
the data. Moreover, the flexibility offered by the choice of the tuning parameter of the kernel, which can be used to
build a trade-off between statistical efficiency and robustness, is another advantage of such estimators.
Thus, it seems natural to apply such inference techniques to copulas, for which the risk of misspecification can sometimes be important.

\mds

In this paper, we will study a general semi-parametric inference procedure for copulas that is robust with respect to corrupted data, and that can be applied in case of model misspecification. Note that other distances are known to induce robustness, like the total variation distance~\cite{Yatracos1985} or the Hellinger distance~\cite{BaraudBirgeSart2017RhoEstimators}. However, the estimation procedures proposed in these papers are not computable. Also, we refer the reader to~\cite{BaraudBirgeSart2017RhoEstimators} for a thorough discussion on why the MLE, based on the Kullback-Leibler divergence, cannot enjoy the same robustness properties.

\mds

\mds

The rest of the paper is organized as follows: the remaining of the introduction yields notations and the definition of our estimators.
Section~\ref{section:theory} contains our theoretical results: non-asymptotic oracle inequalities, consistency and asymptotic distributions of our estimators.
Section~\ref{section:experiments} provides experimental results.
A simulation study confirms the robustness of MMD. We also provide an R Package, called MMDCopula~\cite{MMDCopula}, which allows statisticians to apply our algorithms.

\mds

Note that our package computes the MMD estimator by a stochastic gradient algorithm, described in Section~\ref{section:experiments}. From~\cite{Briol2019,Cherief2019}, such an algorithm can be implemented to compute the MMD estimator as long as it is possible to sample from the model. Thus, our package has been built on the package VineCopula~\cite{VineCopulaR}, which allows to sample from the most popular copula families. This package also provided us some helpful formulas for the densities of some copulas, and their differentials. More details about the implementation can be found in Section~\ref{section:experiments}.

\subsection{Notations}

Let $(\XX_i)_{i=1,\ldots,n}$ be an i.i.d. sample of $d$-dimensional random vectors, whose underlying copula is denoted by $C_0$ and
whose margins are denoted by $F_1,\ldots,F_d$. The latter ones will be left unspecified. We assume these margins are continuous. This standard assumption will allow to invoke powerful results from the theory of empirical copula processes (\cite{Segers2014} in particular).
Let us define the unobservable random variables $U_k=F_k(X_k)$, $k\in\{1,\ldots,d\}$, and $\UU=(U_1,\ldots,U_d)$, for a given random vector  $\XX=(X_1,\ldots,X_d)$ whose underlying copula is $C_0$ and underlying margins are $F_1,\ldots,F_d$. Obviously, the cdf of $\UU$ is $C_0$, whose law is denoted by $\Pb_0$.
The empirical measure associated to $(\XX_i)_{i=1,\ldots,n}$ is denoted as $\Pb_n$.

\mds

We consider a particular parametric family of copulas $\Cc=\{C_\theta, \theta \in \Theta\subset \Rb^p\}$ (the family ``of interest'') and we search
the best-suited copula inside the latter family.
When the model is correctly specified, there exists a ``true'' parameter $\theta_0\in \Theta$ i.e. $C_0=C_{\theta_0}$.
More generally, possibly in case of misspecification, we focus on a  ``pseudo-true'' parameter $\theta^*_0\in \Theta$ so that a particular distance between $C_0$ and
$C_\theta$ is minimized over $\theta\in \Theta$. In our case, this chosen distance will be the MMD.
Denoting by $\Pb^U_\theta$ the law induced by $C_\theta$ on the hypercube $\Uc=[0,1]^d$,
a pseudo-true value is formally defined as
$$  \theta^*_0 \in \arg\min_{\theta\in\Theta} \Db(\Pb^U_\theta,\Pb_0).$$

\mds

In the copula-related literature with unknown margins, it is common to define a pseudo-sample $(\hat \UU_i)_{i=1,\ldots,n}$, where
$ \hat \UU_i=(\hat U_{i,1},\ldots,\hat U_{i,d})$ and
$$ \hat U_{i,k}= F_{n,k}(X_{i,k}),\;\; F_{n,k}(t)= n^{-1}\sum_{i=1}^n \1( X_{i,k}\leq t ),$$
for every $i\in\{1,\ldots,n\}$, $k\in\{1,\ldots,d\}$ and every real number $t$, denoting by $\1(\cdot)$ the usual indicator function.
Our goal will be to evaluate the pseudo-true parameter $\theta^*_0$ with MMD techniques, from the initial sample $(\XX_i)_{i=1,\ldots,n}$ or from the pseudo-sample $(\hat\UU_i)_{i=1,\ldots,n}$.
The empirical distribution of the latter pseudo-sample is called the empirical copula $C_n$ (\cite{JD2004}).
\mds

A relevant idea will be to work on the hypercube $\Uc = [0,1]^d$ instead of $\Rb^d$.
To be specific, imagine we observe $n$ i.i.d. realizations of $\UU$, called $\UU_1,\ldots,\UU_n$, and let $\Pb_n^U$ be the associated empirical measure on $\Uc$.
To obtain an estimator of $\theta$, the MMD criterion to be minimized is then
$ \Db(\Pb^U_\theta,\Pb^U_n)= \| \mu_{\Pb^U_\theta} - \mu_{\Pb^U_n} \|_{\Hc_U},$
for some RKHS $\Hc_U$, that is associated with a kernel $K_U:\Uc\times \Uc \rightarrow \Rb$.
As in~\cite{Briol2019}, we have
\begin{eqnarray*}
\lefteqn{ \Db^2(\Pb^U_\theta,\Pb^U_n) = \int K_U(\uu,\vv) \Pb^U_\theta (d\uu) \, \Pb^U_\theta (d\vv) - 2 \int K_U(\uu,\vv) \Pb^U_\theta (d\uu) \, \Pb^U_n (d\vv) }\\
&+& \int K_U(\uu,\vv) \Pb_n^U (d\uu) \, \Pb_n^U (d\vv). \hspace{6cm}
\end{eqnarray*}
Since we do not observe some realizations of $\UU$, we have to replace them by pseudo-observations in the latter criterion.
This yields the approximate criterion
\begin{eqnarray*}
\lefteqn{ \Db^2(\Pb^U_\theta,\hat\Pb^U_n) =
\int K_U(\uu,\vv) \Pb^U_\theta (d\uu) \, \Pb^U_\theta (d\vv) - 2 \int K_U(\uu,\vv) \Pb^U_\theta (d\uu) \, \hat\Pb^U_n (d\vv) }\\
&+& \int K_U(\uu,\vv) \hat\Pb_n^U (d\uu) \, \hat\Pb_n^U (d\vv), \hspace{6cm}
\end{eqnarray*}
where $   \hat\Pb_n^U $ denotes the empirical measure associated with the pseudo-sample $(\hat \UU_i)_{i=1,\ldots,n}$.
Then, an estimator of $\theta^*_0$ is defined as
\begin{align}
\hat \theta_n
& \in \arg\min_{\theta\in\Theta} \Db(\Pb^U_\theta,\hat\Pb^U_n) \nonumber \\
& \in \arg\min_{\theta\in\Theta}
\int K_U(\uu,\vv) \Pb^U_\theta (d\uu) \, \Pb^U_\theta (d\vv)
- \frac{2}{n}\sum_{i=1}^n \int K_U(\uu,\hat \UU_i) \Pb^U_\theta (d\uu)
\label{RKHS_criterion}
\end{align}
If $C_\theta$ has a density $c_\theta$ with respect to the Lebesgue measure on $[0,1]^d$, this criterion may be rewritten
\begin{equation}
\label{RKHS_criterion:density}
\hat \theta_n \in \arg\min_{\theta\in\Theta}
    \int K_U(\uu,\vv) c_\theta (\uu) c_\theta (\vv) \, d\uu \, d\vv
- \frac{2}{n}\sum_{i=1}^n \int K_U(\uu,\hat \UU_i) c_\theta (\uu) \, d\uu.
\end{equation}

\mds

 It is clear from the definition that $\hat{\theta}_n$ depends on the kernel $K_U$. Thus, the choice of the latter kernel is a very important question. The experimental study in Section~\ref{section:experiments} shows that the most common parametric copulas, Gaussian kernels $K_G(\uu,\vv) = \exp(-\|h(\uu)-h(\vv)\|^2/\gamma^2)$ lead to very good results ($h$ being the identity map or the inverse of the c.d.f of a standard Gaussian random variable, applied coordinatewise). Interestingly, it empirically seems that the value of $\gamma$ that leads to the smallest MSE mainly depends on the kernel, and not really on the sample size nor the true value of the parameter. This is shown in Figure~\ref{FigureGamma}, and in additional plots in the supplementary material. Actually, this fact was rigorously proven in~\cite{Cherief2019} for the Gaussian mean model, and we conjecture that it holds more generally. This allows to calibrate $\gamma$ once and for all through a preliminary set of simulations. Note that Dziugaite et al.~\cite{Dziugaite2015} proposed a median heuristic to calibrate $\gamma$ that yields good results in practice. Alternatively, Briol et al.~\cite{Briol2019} proposed to minimize the asymptotic variance of the estimated parameter, which we could do thanks to our Theorem~\ref{Th_AN_bis}. A more complete discussion on the choice of the kernel can be found page 14 in~\cite{Briol2019}.

\mds

\begin{Remark}
An alternative approach would be to directly work with the initial observations $\XX_i$, instead of the pseudo-observations $\hat\UU_i$.
In this case, we apply the same strategy, but with the initial sample. The ``feasible'' law of $\XX_i$ will be semi-parametric, because its margins are non-parametrically estimated.
To obtain an estimator of $\theta$, the criterion to be minimized would now be
$ \Db(\Pb^X_\theta,\Pb^X_n) = \| \mu_{\Pb^X_\theta} - \mu_{\Pb^X_n} \|_{\Hc_X},$
for some RKHS $\Hc_X$, that is associated with a kernel $K_X:\Rb^d\times \Rb^d \rightarrow \Rb$.
Here, $\Pb^X_\theta$ denotes the law of $\XX$ given by $F_1,\ldots,F_d$ and $C_\theta$. Applying Sklar's theorem, note that, for every $\xx=(x_1,\ldots,x_d)$,
$ \Pb^X_\theta (\XX \leq \xx)=C_{\theta}\big(F_{1}(x_1),\ldots,F_{d}(x_d) \big).$
As above,
\begin{eqnarray*}
\lefteqn{ \Db^2(\Pb^X_\theta,\Pb^X_n) = \int K_X(\xx,\yy) \Pb^X_\theta (d\xx) \, \Pb^X_\theta (d\yy) - 2 \int K_X(\xx,\yy) \Pb^X_\theta (d\xx) \, \Pb^X_n (d\yy) }\\
&+& \int K_X(\xx,\yy) \Pb^X_n (d\xx) \, \Pb^X_n (d\yy).\hspace{6cm}
\end{eqnarray*}
Since we do not know the margins of $\XX$, this yields the approximate criterion
\begin{eqnarray*}
\lefteqn{ \Db^2(\hat\Pb^X_\theta,\Pb^X_n) =
\int K_X(\xx,\yy) \hat\Pb^X_\theta (d\xx) \, \hat\Pb^X_\theta (d\yy) - 2 \int K_X(\xx,\yy) \hat\Pb^X_\theta (d\xx) \, \Pb^X_n (d\yy) }\\
&+& \int K_X(\xx,\yy) \Pb^X_n (d\xx) \, \Pb^X_n (d\yy), \hspace{6cm}
\end{eqnarray*}
where, for every $\xx=(x_1,\ldots,x_d)$, we define $ \hat \Pb^X_\theta (\XX \leq \xx)=C_{\theta}\big(F_{n,1}(x_1),\ldots,F_{n,d}(x_d) \big).$
Then, this provides another estimator
$$ \hat \theta^X_n \in \arg\min_{\theta\in\Theta} \Db(\hat \Pb^X_\theta,\Pb^X_n)
= \arg\min_{\theta\in\Theta}
\int K(\xx,\yy) \hat\Pb^X_\theta (d\xx) \, \hat\Pb^X_\theta (d\yy) - \frac{2}{n}\sum_{i=1}^n \int K(\xx,\XX_i) \hat\Pb^X_\theta (d\xx) .$$
Unfortunately, the evaluation of any integral as $\int \psi(\xx)\, \hat\Pb^X_\theta(d\xx)$ is costly in general.
Indeed,
$$\int \psi(\xx)\, \hat\Pb^X_\theta(d\xx) \simeq  n^{-d}\sum_{i_1,\ldots,i_d=1}^{n} \psi(X_{i_1,1},\ldots,X_{i_d,d})
c_\theta\big( F_{n,1}(X_{i_1,1}),\ldots, F_{n,d}(X_{i_d,d}) \big).$$
Therefore, it is more convenient to deal with the first method, especially if $d$ is large. This is our choice in this paper.
\end{Remark}

\mds


\section{Theoretical results}
\label{section:theory}

We now study the theoretical properties of the estimator defined by~(\ref{RKHS_criterion}).
Since we will work with pseudo-observations from now on, we omit the upper index ``$U$'' to lighten notations.
Thus, the law induced by the pseudo-sample $(\hat\UU_i)_{i=1,\ldots,n}$, previously denoted $\hat\Pb_n^U$, simply becomes  $\hat\Pb_n$.
Moreover, $\Pb_n^U$, the law of the unobservable sample $(\UU_i)_{i=1,\ldots,n}$ becomes $\Pb_n$. Recall that the true underlying law is $\Pb_0$,
and $\Pb_0=\Pb_{\theta_0^*}$ only if the model is correctly specified.
For any function $f:\Ec\subset \Rb^d\rightarrow \Rb$ that is twice continuously differentiable, set
$$ \|d^{(2)} f\|_\infty=\sup_{\xx\in \Ec} \sup_{k,l=1,\ldots,d} \bigg|\frac{ \partial^2 f}{\partial x_k \partial x_l} (\xx) \bigg|.$$

We assume in this section that the kernel $K_U$ is symmetrical, i.e. $K_U(\uu,\vv)=K_U(\vv,\uu)$ for every $\uu$ and $\vv$ in $[0,1]^d$ (otherwise, replace $K_U$ by a symmetrized version). We also assume that the kernel is bounded over $[0,1]^2$. Note that the popular Gaussian kernel $K_G(\uu,\vv) = \exp(-\|\uu-\vv\|^2/\gamma^2)$, is characteristic, symmetric and bounded.
We recall that, when $K$ is a characteristic kernel, the divergence
\begin{eqnarray*}
\lefteqn{ \Db^2(\Pb,\Qb) = \int K_U(\uu,\vv) \Pb (d\uu) \, \Pb (d\vv) - 2 \int K_U(\uu,\vv) \Pb (d\uu) \, \Qb (d\vv) }\\
&+& \int K_U(\uu,\vv) \Qb (d\uu) \, \Qb (d\vv), \hspace{6cm}
\end{eqnarray*}
induces a true distance between probability measures on $[0,1]^d$.

\subsection{Non-asymptotic guarantees}


The first result of this section is a non-asymptotic ``universal'' upper bound in terms of MMD distance that holds with high probability for any underlying distribution.
Our bound exhibits clear dimensionality- and kernel-dependent constants. It establishes that the MMD estimator is robust to misspecification, and is consistent at the usual optimal $n^{-1/2}$ rate. Similar results can be found in the literature, both in the i.i.d.\ (Theorem 1 in \cite{Briol2019}, Theorem 3.1 in \cite{Cherief2019}) and in the dependent setting (Theorem 3.2 in \cite{Cherief2019}), but none of them can be applied to semi-parametric copula models.

\begin{Theorem}
\label{upper_bound_finite_distance}
The kernel $K_U$ is assumed to be two times continuously differentiable on $[0,1]^d$.
Then for any $\nu,\delta>0$ with $\nu+\delta<1$, with probability larger than $1-\delta-\nu \in (0,1)$,
\begin{align*}
\Db(\Pb_{\hat\theta_n},\Pb_0)
\leq & \inf_{\theta\in\Theta} \Db(\Pb_{\theta},\Pb_0) +
\bigg\{ \frac{8}{n}\sup_{\uu\in [0,1]^d} K_U(\uu,\uu)  \bigg\}^{1/2} \Big\{ 1 + \big(-\ln \delta \big)^{1/2} \Big\} \\
& \quad + \; \Bigg\{\frac{2d^2}{n} \|d^{(2)} K_U\|_\infty \ln\bigg(\frac{2d}{\nu}\bigg)\Bigg\}^{1/2}.\hspace{6cm}
\end{align*}
\end{Theorem}
Note that, if a pseudo-true value $\theta_0^*$ exists, $\inf_{\theta\in\Theta} \Db(\Pb_{\theta},\Pb_0)=\Db(\Pb_{\theta_0^*},\Pb_0)$ by definition, and this quantity is zero if the model is correctly specified.

\begin{proof}
For every $\theta\in \Theta$, we have
\begin{eqnarray*}
\lefteqn{ \Db(\Pb_{\hat\theta_n},\Pb_0) \leq \Db(\Pb_{\hat\theta_n},\hat\Pb_n) + \Db (\hat\Pb_n,\Pb_n) + \Db (\Pb_n,\Pb_0)  }\\
&\leq & \Db(\Pb_{\theta},\hat\Pb_n)+ \Db (\hat\Pb_n,\Pb_n) + \Db (\Pb_n,\Pb_0) \\
&\leq & \Db(\Pb_{\theta},\Pb_0)+ 2\Db (\hat\Pb_n,\Pb_n) + 2\Db (\Pb_n,\Pb_0).
\end{eqnarray*}
With probability greater than $1-\delta$, Lemma 1 in~\cite{Briol2019} yields
\begin{equation}
\Db (\Pb_n,\Pb_0) \leq \bigg\{ \frac{2}{n}\sup_{\uu\in [0,1]^d} K_U(\uu,\uu) \bigg\}^{1/2} \Big\{ 1+ \big(-\ln\delta\big)^{1/2} \Big\}.
\label{T1_distance_finite}
\end{equation}
Moreover, by some limited expansions of $K_U$ wrt each of its arguments, evaluated at $(\UU_i,\UU_j)$ and with matrix notations, we get
\begin{eqnarray*}
\lefteqn{ \Db^2 (\hat\Pb_n,\Pb_n) =
\frac{1}{n^2} \sum_{i,j=1}^n \Big\{ K_U( \UU_i,\UU_j)-2K_U(\hat \UU_i,\UU_j)  +K_U(\hat \UU_i,\hat\UU_j)     \Big\}  }\\
&= &
\frac{1}{n^2} \sum_{i,j=1}^n \Big\{ \partial_1 K_U( \UU_i,\UU_j)^\top (\UU_i - \hat\UU_i) -  \frac{1}{2}(\hat \UU_i - \UU_i)^\top \partial^2_1 K_U( \UU^*_i,\UU_j) (\hat \UU_i - \UU_i) \\
&-&
\partial_2 K_U( \hat\UU_i,\UU_j)^\top ( \UU_j - \hat\UU_j) +  \frac{1}{2}(\hat \UU_j - \UU_j)^\top \partial^2_2 K_U( \hat\UU_i,\tilde\UU_j) (\hat \UU_j - \UU_j) \Big\},
\end{eqnarray*}
for some random vectors $\UU_i^*$ (resp. $\tilde \UU_j$) that lie between $\UU_i$ and $\hat\UU_i$ (resp. between $\UU_j$ and $\hat\UU_j$).
Since the kernel is symmetrical, $ \partial_1 K_U( \uu,\vv) = \partial_2 K_U(\vv,\uu)$ for every $(\uu,\vv)$ in $[0,1]^{2d}$.
This yields, with obvious notations,
\begin{eqnarray*}
\lefteqn{ \Db^2 (\hat\Pb_n,\Pb_n) =
\frac{1}{n^2} \sum_{i,j=1}^n \Big\{ \frac{(-1)}{2}(\hat \UU_i - \UU_i)^\top \partial^2_1 K_U( \UU^*_i,\UU_j) (\hat \UU_i - \UU_i) \nonumber }\\
&-&
(\hat \UU_i - \UU_i)^\top \partial^2_{12} K_U( \bar\UU_i,\UU_j) ( \UU_j - \hat\UU_j) +  \frac{1}{2}(\hat \UU_j - \UU_j)^\top \partial^2_2 K_U( \hat\UU_i,\tilde\UU_j) (\hat \UU_j - \UU_j) \Big\},
\label{Limited_expansion_1}
\end{eqnarray*}
and we deduce
\begin{equation*}
\Db^2 (\hat\Pb_n,\Pb_n)\leq 2d^2 \|d^{(2)} K_U\|_\infty \sup_{i=1,\ldots,n} \sup_{k=1,\ldots,d} |\hat U_{ik} - U_{ik}|^2 .
\end{equation*}

The Dvoretzky-Kiefer-Wolfowitz inequality (p. 383 in~\cite{Boucheron2012}) yields
$$ \Pb\Big( \sup_{i=1,\ldots,n} \sup_{k=1,\ldots,d} |\hat U_{i,k} - U_{i,k}|^2 > \eps  \Big) \leq 2d \exp \big(- 2n\eps \big),$$
and $\Db^2 (\hat\Pb_n,\Pb_n)$ is less than $d^2 \|d^{(2)} K_U\|_\infty \ln(2d/\nu)/n$ with a probability larger than $1-\nu$.
In addition with~(\ref{T1_distance_finite}), this proves the result.
\end{proof}

\begin{Remark}
Note that if an exact minimizer $\hat\theta_n$ of \eqref{RKHS_criterion} does not exist, we can simply define $\hat \theta_n$ as any value that reaches the infimum up to $1/n$. The extension of Theorem 1 to this case is direct.
\end{Remark}

It is possible to slightly strengthen Theorem~\ref{upper_bound_finite_distance} at the price of more regularity for $K_U$, details are provided in Appendix A.

 Let us emphasize the consequences of Theorem~\ref{upper_bound_finite_distance} when the data are contaminated by a proportion $\varepsilon$ of outliers. Huber proposed a contamination model for which $\mathbb{P}_0 = (1-\varepsilon) \mathbb{P}_{\theta_0} + \varepsilon \mathbb{Q}$. That is, while the majority of the observations is actually generated from the ``true'' model, a (small) proportion $\varepsilon$ of them is generated by an arbitrary contamination distribution $\mathbb{Q}$. Using this framework, it is possible to upper bound the distance between the MMD estimator and the true parameter directly. To be short, assume here that
 $ \sup_{\uu\in [0,1]^d} K_U(\uu,\uu) \leq 1 $,
 as for the usual Gaussian kernel.
 Since $ \mathbb{D}(\mathbb{P}_0,\mathbb{P}_{\theta_0}) \leq 2\varepsilon$ and $  \Db(\Pb_{\hat\theta_n},\Pb_{\theta_0})  \leq 2\varepsilon +  \Db(\Pb_{\hat\theta_n},\mathbb{P}_{0})$ by the triangle inequality, Theorem~\ref{upper_bound_finite_distance} yields
\begin{equation}
\label{bound:robustness}
 \Db(\Pb_{\hat\theta_n},\Pb_{\theta_0}) \leq 4\varepsilon +
  \left(\frac{8}{n}\right)^{\frac{1}{2}} \Big\{ 1+ \big(-\ln \delta \big)^{1/2} \Big\}
+ \bigg\{\frac{2d^2}{n} \|d^{(2)} K_U\|_\infty \ln\big(\frac{2d}{\nu}\big)\bigg\}^{1/2}.
\end{equation}
In any model where an upper bound on $\|\hat{\theta}_n-\theta_0\|^2$ can be deduced from an upper bound on $\Db(\Pb_{\hat\theta_n},\Pb_{\theta_0})$, this proves the robustness of $\hat{\theta}_n$.

\begin{Example}
\label{exm:Gaussian:robustness}
 As an illustration, let us consider the Gaussian copula model in dimension $d=2$, whose laws $(\Pb_{\theta})_{\theta\in(-1,1)}$ are given by their density
\begin{equation}
 c_\theta(u_1,u_2)= \frac{1}{2\pi\sqrt{1-\theta^2}\phi(x_1) \phi(x_2)} \exp\Big(-\frac{1}{2(1-\theta^2)} \big( x_1^2+x_2^2 - 2\theta x_1 x_2  \big)   \Big),
\label{def_cop_gauss}
\end{equation}
by setting $ x_k=\Phi^{-1}(u_k)$, $k=1,2$. We use the Gaussian kernel:
$$
K_U(\UU,\VV) = \exp\left\{ -\|\Phi^{-1}(\UU)-\Phi^{-1}(\VV)\|^2/\gamma^2 \right\}, $$
where $\Phi$ is the c.d.f of a standard Gaussian random variable, and its inverse $\Phi^{-1}$ is applied coordinatewise. We prove at the end of Appendix F 
that, using the latter Gaussian kernel, there is a constant $c(\gamma)\in(0,+\infty)$ that depends only on $\gamma$ such that, for any $(\theta_1,\theta_2)\in(-1,1)^2$,
$|\theta_1-\theta_2| \leq c(\gamma) \Db(\Pb_{\theta_1},\Pb_{\theta_2}).$
Together with~\eqref{bound:robustness}, this gives:
\begin{equation*}
|\hat{\theta}_n-\theta_0| \leq c(\gamma) \left[ 4 \varepsilon +
   \left(\frac{8}{n}\right)^{\frac{1}{2}} \Big\{ 1+ \big(-\ln \delta \big)^{1/2} \Big\}
+ \bigg(\frac{8}{n} \|d^{(2)} K_U\|_\infty \ln\big(\frac{4}{\nu}\big)\bigg)^{1/2}\right].
\end{equation*}
\end{Example}

In the general case, we can use the following proposition.
\begin{Proposition}
\label{prop_TAYLOR}
 Assume that the map $\theta\mapsto  \Db^2(\Pb_{\theta},\Pb_{\theta_0})$ is
 twice continuously differentiable in a neighborhood of $\theta_0$.
Denoting by $\lambda_{\min}(\theta)$ the smallest eigenvalue of $\nabla^2_{\theta,\theta} \Db^2(\Pb_{\theta},\Pb_{\theta_0})$, assume that
$\lambda_{\min}(\theta)\geq \lambda_{\min}(\theta_0)/2>0$ when $\|\theta - \theta_0\|<r$, for some $r>0$.
Set $\alpha=\inf_{\{\theta; \|\theta - \theta_0\|\geq r\}} \Db(\Pb_{\theta},\Pb_{\theta_0})$ and assume $\alpha >0$.\\
Then, for any contamination distribution $\Qb$, when the data are drawn from $(1-\varepsilon)\Pb_{\theta_0} + \varepsilon \Qb$ for some $\varepsilon\in[0,\alpha/8]$, for any $\nu>0$ and $\delta>0$ with $\nu+\delta<1$, as soon as
 $$ \sqrt{n}\alpha \geq
 \bigg\{ 32\sup_{\uu\in [0,1]^d} K_U(\uu,\uu)  \bigg\}^{1/2} \Big\{ 1 + \big(-\ln \delta \big)^{1/2} \Big\}
  +\bigg\{ 8 \,d^2 \|d^{(2)} K_U\|_\infty \ln\Big(\frac{2d}{\nu}\Big)\bigg\}^{1/2},$$
 we have, with probability at least $1-\nu-\delta$,
 \begin{equation*}
\|\hat{\theta}_n-\theta_0\| \leq \frac{2}{\sqrt{\lambda_{\min}(\theta_0)}} \bigg[ 4 \varepsilon +
   \Big\{\frac{8}{n}\sup_{\uu\in [0,1]^d} K_U(\uu,\uu) \Big\}^{\frac{1}{2}} \Big\{ 1+ \big(-\ln \delta \big)^{\frac{1}{2}} \Big\}
+ \Big\{\frac{2 d^2}{n} \|d^{(2)} K_U\|_\infty \ln\big(\frac{2d}{\nu}\big)\Big\}^{\frac{1}{2}}\bigg].
\end{equation*}
\end{Proposition}
The proof is provided in Appendix B
.

\subsection{Asymptotic guarantees}
\label{Asymptotic_results}


We denote
\begin{align*}
\ell(\ww;\theta) =
\int K_U(\uu,\vv) \Pb_\theta (d\uu) \, \Pb_\theta (d\vv) - 2 \int K_U(\uu,\ww) \Pb_\theta (d\uu) .
\end{align*}
We assume that the functions $\ell(\cdot;\theta)$ are measurable and $\Pb_0$-integrable for every $\theta\in\Theta$.
The theoretical loss function is
$$
L_0(\theta)=\mathbb{E}[ \ell(\UU;\theta) ] = \int_{[0,1]^d} \ell(\ww;\theta) \Pb_0(d\ww) .
$$
Here, it is approximated by the empirical ``feasible'' loss
$$L_n(\theta)=\frac{1}{n}\sum_{i=1}^n \ell(\hat \UU_i;\theta) = \int_{[0,1]^d} \ell(\ww;\theta) \hat{\Pb}_n(d\ww) , $$
so that $\hat{\theta}_n \in \arg\min_{\theta\in\Theta} L_n(\theta)$ and $\theta^*_0 \in \arg\min_{\theta\in\Theta} L_0(\theta)$.
The asymptotic properties of M-estimators (``Quasi-MLE'' particularly) for possibly misspecified models are well established in the literature: see~\cite{White82, White94} for instance.
As usual in the statistical theory of copulas, the main difficulty will come here from unspecified margins.

\subsubsection{Consistency}

Under classical assumptions, we prove that the MMD estimator is consistent.

\begin{Condition}
The parameter space $\Theta$ is compact.
The map $L_0:\Theta\rightarrow \Rb$ is continuous on $\Theta$ and uniquely minimized at $\theta_0^*$.
\label{cond_consistency1}
\end{Condition}

\begin{Condition}
The family $\Fc=\{ \ell(\cdot,\theta);\, \theta\in \Theta\}$ is a collection of measurable functions with an integrable envelope function $F$.
For every $\ww\in [0,1]^d$, the map $\theta \mapsto \ell(\ww;\theta)$ is continuous on $\Theta$.
\label{cond_consistency2}
\end{Condition}

\begin{Theorem}
\label{thm_consistency}
If Conditions \ref{cond_consistency1} and \ref{cond_consistency2} are fulfilled, then $\hat{\theta}_n$ is strongly consistent, i.e.
$$\hat{\theta}_n \xrightarrow[n\rightarrow+\infty]{\Pb_0 - a.s.} \theta_0^* .$$
\end{Theorem}

\begin{proof}
As $\Theta$ is compact, then the $\delta$-bracketing numbers $\mathcal{N}_{[\cdot]}\big(\delta,\Fc,L^1(\Pb_0)\big)$ are finite for every $\delta>0$, invoking Example 19.8 in~\cite{VdV98}.
Moreover, using Lemma 1(c) in~\cite{Chen2005}, we obtain the strong uniform law of large numbers 
$$
\sup_{\theta\in\Theta} |L_0(\theta) - L_n(\theta)| \xrightarrow[n\rightarrow+\infty]{\Pb_0 - a.s.} 0 .
$$
Hence, according to Theorem 2.1 in~\cite{NeweyMcFadden} for example, we deduce the strong consistency of the minimizer $\hat{\theta}_n$ of $L_n$ towards the unique minimizer of $L_0$.
\end{proof}

\subsubsection{Asymptotic normality}

Although Theorem \ref{thm_consistency} gives conditions under which we obtain the consistency of the MMD estimator, it does not provide any information on its rate of convergence. Hence, we now state the weak convergence of $\sqrt{n}(\hat\theta_n-\theta_0^*)$.
First, we need a set of usual regularity conditions to deal with M-estimators. It mainly requires the functions $\ell(\ww;\cdot)$ to be smooth enough on a small neighborhood of $\theta_0^*$ when
$\ww\in[0,1]^d$.

\begin{Condition}
$\theta_0^*$ is an interior point of $\Theta$.
\label{cond_normality_0}
\end{Condition}


\begin{Condition}
There exists an open neighborhood $\mathcal{O}\subset\Theta$ of $\theta_0^*$ such that the maps $\theta \mapsto \ell(\ww;\theta)$ are twice continuously differentiable on $\Oc$, for $\Pb_0$-almost every $\ww\in[0,1]^d$.
 Moreover, all functions $\nabla^2_{\theta,\theta}\ell(\cdot;\theta)$ are measurable on $[0,1]^d$ for any $\theta\in\mathcal{O}$.
 \label{cond_normality_1}
\end{Condition}

\begin{Condition}
There exists a compact set $K_0\subset\mathcal{O}$ whose interior contains $\theta_0^*$ such that 
$$\Eb\bigg[  \sup_{\theta\in K_0} \big\|\nabla_{\theta,\theta}^2\ell(\UU;\theta) \big\| \bigg] < +\infty,$$
for any  matrix norm $\|\cdot\|$.
Moreover, the map $\theta\mapsto \mathbb{E}[\nabla^2_{\theta,\theta}\ell(\UU;\theta)]$ is continuous at $\theta_0^*$.
\label{cond_normality_2}
\end{Condition}

\begin{Condition}
The matrix $B=\mathbb{E}[\nabla_{\theta,\theta}^2 \ell(\UU;\theta_0^*) ]$ is positive definite.
\label{cond_normality_3}
\end{Condition}

\begin{Condition}
$\mathbb{E}[\nabla_\theta\ell(\UU;\theta_0^*)]=0$.
\label{cond_normality_plus}
\end{Condition}





Second, the asymptotic behavior of our estimator is closely related to the asymptotic distribution of the empirical copula that has been widely studied in the last two decades. 
The weak convergence in $(\ell^{\infty}([0,1]^d),\|\cdot\|_{\infty})$ of the empirical copula process $\{\sqrt{n}(\hat\Pb_n-\Pb_0)(\uu), \uu \in [0,1]^d \}$ to a Gaussian process was formally stated by \cite{JD2004}, by requiring the first-order partial derivatives of the copula $\Pb_0$ to exist and to be continuous on the entire unit hypercube $[0,1]^d$.
Actually, as initially suggested in Theorem 4 of~\cite{JD2004}, the continuity is not needed on the boundary of the hypercube, but only on the interior of the hypercube. This result was established by \cite{Segers2012} under minimal assumptions, rewritten below as Condition~\ref{cond_normality_4_bis}.
With additional smoothness requirements on the loss function $\ell$ (Condition~\ref{cond_normality_5_bis}), we will be able to obtain the asymptotic normality of our MMD estimator $\hat{\theta}_n$ from the weak convergence of the empirical copula process.

\begin{Condition}
The function $\nabla_{\theta}\ell(\cdot;\theta_0^*)$ is right continuous, i.e. it is coordinatewise right-continuous in each coordinate, and is of bounded variation in the sense of Hardy-Krause (see~\cite{Radoluvic}, Section 2).
\label{cond_normality_5_bis}
\end{Condition}

\begin{Condition}
For each $j\in \{1,\ldots,d\}$, the $j$-th first-order partial derivative $\dot{C}_j$ of the true copula $\Pb_0$ exists and is continuous on the set $V_j=\{\ww\in[0,1]^d: 0< w_j<1\}$.
\label{cond_normality_4_bis}
\end{Condition}

Still, it is possible to obtain the weak convergence of the empirical copula process for an even larger class of copulas using semi-metrics on $\ell^{\infty}([0,1]^d)$ that are weaker than the sup-norm, but the limiting distribution will no longer be Gaussian in general. Indeed, Segers~\cite{Segers2014} established the hypi-convergence of the empirical copula process $\{\sqrt{n}(\hat\Pb_n-\Pb_0)(\uu), \uu \in [0,1]^d \}$ under the following assumption that is weaker than Condition~\ref{cond_normality_4_bis}.


\begin{Condition}
The set $\mathcal{S}$ of points in $[0,1]^d$ where the partial derivatives of the true copula $\Pb_0$ exist and are continuous has Lebesgue measure $1$.
\label{cond_normality_4_ter}
\end{Condition}

Note a related regularity assumption in~\cite{Genest2017}, Condition 1.
Hereafter, $(\ww_I,\1_{-I})$ denotes a vector in $[0,1]^d$ whose $j$-th component is $w_j$ when $j\in I$ and is one otherwise.

\begin{Condition}
For any $I\subset \{1,\ldots,d\}$, $I\neq \emptyset$, there exists some $q_I\in (1,+\infty)$ such that
 $ \int_{[0,1]^{|I|}}  \big| \nabla_\theta\ell(d\ww_I,\1_{-I};\theta_0^*) \big|^{q_I} < \infty.$
\label{cond_moment_dell}
\end{Condition}

Now, let us state the weak convergence of $\sqrt{n}(\hat\theta_n - \theta_0^*)$.
\begin{Theorem}
\label{Th_AN}
If Conditions \ref{cond_consistency1}-\ref{cond_normality_4_bis}
are fulfilled, then $\sqrt{n}(\hat\theta_n - \theta_0^*)$ is asymptotically normal.
Alternatively, under Conditions \ref{cond_consistency1}-\ref{cond_normality_5_bis} and \ref{cond_normality_4_ter}-\ref{cond_moment_dell},
the weak limit of $\sqrt{n}\big(\hat{\theta}_n - \theta_0^*\big)$ still exists.
\end{Theorem}

\begin{proof}
According to Condition~\ref{cond_normality_1}, $L_n$ is twice differentiable on a neighborhood of $\theta_0^*$ and
$\partial L_n/\partial\theta_j = n^{-1} \sum_{i=1}^n \partial\ell(\hat \UU_i;\cdot)/\partial\theta_j$.
Moreover, due to the the consistency of $\hat\theta_n$ (according to Conditions \ref{cond_consistency1} and \ref{cond_consistency2}), we can assume that
$\hat\theta_n$ belongs to such a neighborhood. Using Condition \ref{cond_normality_0}, the first-order condition is
\begin{equation}
    0 =  \nabla_\theta L_n(\hat\theta_n) = \nabla_\theta L_n(\theta_0^*) + \nabla_{\theta,\theta^\top} L_n(\bar\theta_n) (\hat\theta_n-\theta_0^*),
\label{matrix_equation}
\end{equation}
where $\bar\theta_{n}$ is a random vector whose components lie between those of $\theta_0^*$ and $\hat\theta_n$.
Note that $H_n=\nabla_{\theta,\theta^\top} L_n(\bar\theta_n)$ is an $(d,d)$-sized Hessian matrix whose $(j,k)$-th component is
$ H_{n,jk}=\frac{1}{n}\sum_{i=1}^n \partial^2\ell(\hat \UU_i;\bar\theta_{n})/\partial\theta_k\partial\theta_j$, $j,k \in \{1,\ldots,d\}$
Let us now study the asymptotic behavior of this Hessian matrix and of $\nabla_\theta L_n(\theta_0^*)$.

\mds

For any pair $(j,k)$, the function $\partial^2\ell(\ww;\cdot)/\partial\theta_j\partial\theta_k$ is continuous on the compact set $K$ for $\Pb_0$ almost every $\ww\in[0,1]^d$, all second-order functions $\partial^2\ell(\cdot;\theta)/\partial\theta_j\partial\theta_k$ are measurable for any $\theta\in K$ and $\mathbb{E}[\sup_{\theta\in K}|\partial^2\ell(\UU;\theta)/\partial\theta_k\partial\theta_j|]<+\infty$
(Conditions \ref{cond_normality_1} and \ref{cond_normality_2}).
Therefore, the $L^1$ bracketing numbers associated to the Hessian maps indexed by $\theta \in K$ are finite, invoking Example 19.8 in~\cite{VdV98}.
Using Lemma 1(c) in~\cite{Chen2005}, we get 
$$
\sup_{\theta\in K} \bigg|\frac{1}{n} \sum_{i=1}^n \frac{\partial^2\ell(\hat\UU_i;\theta)}{\partial\theta_k\partial\theta_j} - \mathbb{E}\bigg[\frac{\partial^2\ell(\UU;\theta)}{\partial\theta_k\partial\theta_j}\bigg]\bigg| \xrightarrow[n\rightarrow+\infty]{\Pb_0 - a.s.} 0 .
$$
As $\bar\theta_{n}$ lies between $\hat\theta_n$ and $\theta_0^*$ componentwise, $\bar\theta_{n} \xrightarrow[n\rightarrow+\infty]{\Pb_0 - a.s.} \theta_0^*$.
Moreover, taking expectations with respect to $(\UU, \bar \theta_n)$ or $(\hat\UU, \bar \theta_n)$ respectively, we have for $n$ large enough
\begin{align*}
\bigg| & \frac{1}{n}\sum_{i=1}^n \frac{\partial^2\ell(\hat \UU_i;\bar\theta_{n})}{\partial\theta_k\partial\theta_j} - \mathbb{E}\bigg[\frac{\partial^2\ell(\UU;\theta_0^*)}{\partial\theta_k\partial\theta_j}\bigg] \bigg| \\
& \leq \bigg| \frac{1}{n}\sum_{i=1}^n \frac{\partial^2\ell(\hat \UU_i;\bar\theta_{n})}{\partial\theta_k\partial\theta_j} - \mathbb{E}\bigg[\frac{\partial^2\ell(\UU;\bar\theta_{n})}{\partial\theta_k\partial\theta_j}\bigg] \bigg| + \bigg| \mathbb{E}\bigg[\frac{\partial^2\ell(\UU;\bar\theta_{n})}{\partial\theta_k\partial\theta_j}\bigg] - \mathbb{E}\bigg[\frac{\partial^2\ell(\UU;\theta_0^*)}{\partial\theta_k\partial\theta_j}\bigg] \bigg|  \\
& \leq \, \sup_{\theta\in K} \, \bigg|\frac{1}{n} \sum_{i=1}^n \frac{\partial^2\ell(\hat\UU_i;\theta)}{\partial\theta_k\partial\theta_j} - \mathbb{E}\bigg[\frac{\partial^2\ell(\UU;\theta)}{\partial\theta_k\partial\theta_j}\bigg]\bigg|  + \bigg| \mathbb{E}\bigg[\frac{\partial^2\ell(\UU;\bar\theta_{n})}{\partial\theta_k\partial\theta_j}\bigg] - \mathbb{E}\bigg[\frac{\partial^2\ell(\UU;\theta_0^*)}{\partial\theta_k\partial\theta_j}\bigg] \bigg|.
\end{align*}
The continuity of $\mathbb{E}[\partial^2\ell(\UU;\cdot)/\partial\theta_j\partial\theta_k]$ at $\theta_0^*$ (Condition \ref{cond_normality_1}) yields
$$
\frac{1}{n}\sum_{i=1}^n \frac{\partial^2\ell(\hat \UU_i;\bar\theta_{n})}{\partial\theta_k\partial\theta_j} \xrightarrow[n\rightarrow+\infty]{\Pb_0 - a.s.} \mathbb{E}\bigg[\frac{\partial^2\ell(\UU;\theta_0^*)}{\theta_j\theta_k}\bigg]  .
$$
Finally, by definition of $H_n$ and $B$ (see Condition \ref{cond_normality_3}), we obtain $H_n \xrightarrow[n\rightarrow+\infty]{\Pb_0 - a.s.} B$.

\mds

According to Proposition 3.1 in \cite{Segers2012} and under Condition \ref{cond_normality_4_bis}, the empirical copula process $\sqrt{n} (\hat{\Pb}_n-\Pb_0)$ weakly converges to the Gaussian process $\alpha(\ww)-\sum_{j=1}^d\dot{C}_j(\ww)\alpha_j(\ww_j)$ in $\ell^{\infty}([0,1]^d)$ where $\alpha$ is a $\Pb_0$-Brownian bridge.
By Condition \ref{cond_normality_5_bis} and an integration by parts argument (Proposition 3 in~\cite{Radoluvic}), we have with obvious notations
\begin{eqnarray}
\lefteqn{ \sqrt{n}\big\{\nabla_\theta L_n(\theta_0^*)-\mathbb{E}[\nabla_\theta\ell(\UU;\theta_0^*)] \big\}
=\sqrt{n}\int_{(0,1]^d} \nabla_\theta\ell(\ww;\theta_0^*) d(\hat{\Pb}_n-\Pb_0)(\ww)  \nonumber}\\
&=&
\sum_{I\subset \{1,\ldots,d\};I\neq \emptyset}
(-1)^{|I|} \int_{(\0_I,\1_I]}  \sqrt{n}(\hat{\Pb}_n-\Pb_0)(\ww_I,\1_{-I})\, \nabla_\theta\ell(d\ww_I,\1_{-I};\theta_0^*).\hspace{5cm}
\label{Ibyparts}
\end{eqnarray}
Since all the maps $\ww_I\mapsto \nabla_\theta\ell(\ww_I,\1_{-I};\theta_0^*)$ are of bounded variation, the maps
$$ g\mapsto \int_{(\0_I,\1_I]} g(\ww_I)\, \nabla_\theta\ell(d\ww_I,\1_{-I};\theta_0^*)$$
are continuous on $\ell^{\infty}\big([0,1]^{|I|},\|\cdot\|_\infty\big)$ for any $I\neq \emptyset$.
Recalling Condition \ref{cond_normality_plus}, the continuous mapping theorem implies that the weak limit of
$\sqrt{n}\nabla_\theta L_n(\theta_0^*)$ exists, is centered and Gaussian:
\begin{eqnarray*}
\lefteqn{
\sqrt{n}\nabla_\theta L_n(\theta_0^*) \xrightarrow[n\rightarrow+\infty]{\mathcal{L}} }\\
&&
\sum_{I\subset \{1,\ldots,d\};I\neq \emptyset}
(-1)^{|I|} \int_{(\0_I,\1_I]}  \Big\{\alpha(\ww_I,\1_{-I})-\sum_{j\in I}\dot{C}_j(\ww_I,\1_{-I})\alpha_j(\ww_j)\Big\} \,
\nabla_\theta\ell(d\ww_I,\1_{-I};\theta_0^*).
\end{eqnarray*}
Invoking the integration by parts again, this yields
$$
\sqrt{n}\nabla_\theta L_n(\theta_0^*) \xrightarrow[n\rightarrow+\infty]{\mathcal{L}} \int \nabla_{\theta}\ell(\ww;\theta_0^*) d\Big\{\alpha(\ww)-\sum_{j\in I}\dot{C}_j(\ww)\alpha_j(\ww_j)\Big\}  .
$$
As the limiting matrix $B$ is invertible, we can infer that the matrix $H_n$ is a.s. invertible for a sufficiently large $n$.
Using Slutsky's lemma and Formula \eqref{matrix_equation}, we get
$$
\sqrt{n}(\hat\theta_n - \theta_0^*) = H_n^{-1} \sqrt{n}\nabla_\theta L_n(\theta_0^*) \xrightarrow[n\rightarrow+\infty]{\mathcal{L}} B^{-1} \int
\nabla_{\theta}\ell(\ww;\theta_0^*)\,d\Big\{\alpha(\ww)-\sum_{j=1}^d\dot{C}_j(\ww)\alpha_j(\ww_j)\Big\}.
$$

\mds
If Condition \ref{cond_normality_4_bis} is replaced by Condition \ref{cond_normality_4_ter}, then the empirical process $\sqrt{n} (\hat{\Pb}_n-\Pb_0)$ weakly converges to the process $\alpha(\ww) + dC_{(-\alpha_1,...,-\alpha_d)}(\ww)$ in $L_p([0,1]^d)$ for any $1\leq p < \infty$, as detailed in~\cite{Segers2014} (Theorem 4.5. and the remarks that follow).
Due to Condition~\ref{cond_moment_dell} and H\"older's inequality, the maps $h \rightarrow \int h(\ww_I) \,\nabla_\theta\ell(d\ww_I,\1_{-I};\theta_0^*) $ are continuous on $L_{p_I}([0,1]^{|I|})$, $1/p_I+1/q_I=1$.
Therefore, by~(\ref{Ibyparts}) and the continuous mapping theorem, the weak limit of
$\sqrt{n}\big\{\nabla_\theta L_n(\theta_0^*)-\mathbb{E}[\nabla_\theta\ell(\UU;\theta_0^*)]\big\}$
exists and is $B^{-1} \int \nabla_{\theta}\ell(\ww;\theta_0^*) \,d\big\{\alpha(\ww) + dC_{(-\alpha_1,...,-\alpha_d)}(\ww)\big\}  $, proving the result.
\end{proof}

\mds
In the case of asymptotic normality, the asymptotic variance of $\sqrt{n}(\hat\theta_n - \theta_0^*)$ is $B^{-1} \Sigma B^{-1}$, where
$$ \Sigma = \int  \nabla_{\theta}\ell(\ww;\theta_0^*)   \nabla_{\theta}\ell(\ww';\theta_0^*)^\top \,\Cc_0(d\ww,d\ww') ,$$
and $\Cc_0(\cdot,\cdot)$ is the covariance function associated to the limiting law of the empirical copula process, i.e.
$$ \Cc_0(\ww,\ww')= \Eb\Big[
\big\{\alpha(\ww)-\sum_{j=1}^d\dot{C}_j(\ww)\alpha_j(w_j)\big\}  \big\{\alpha(\ww')-\sum_{j=1}^d\dot{C}_j(\ww')\alpha_j(w'_j)\big\}  \Big],$$
denoting by $\alpha$ a usual $\Pb_0$-Brownian bridge on $[0,1]^d$. In particular, note that
$$
\Eb[\alpha(\ww)\alpha(\ww')]=C_0(\ww \wedge \ww') - C_0(\ww)C_0(\ww'),\; (\ww,\ww')\in [0,1]^{2d},
$$
denoting $\ww \wedge \ww'=\big(\min(w_1,w_1'),\ldots,\min(w_d,w_d')\big)$.
The previous matrices can be empirically estimated: see Remark 2 in~\cite{Chen2005}, or~\cite{tsukahara2005}.
Note that a more explicit formula of $\Sigma$ is given in the latter papers, say
\begin{align}
\Sigma=  \, \text{Var}\Big[ \nabla_{\theta}\ell(\UU;\theta_0^*) + \sum_{j=1}^d \int \nabla^2_{\theta,u_j}\ell(\uu;\theta_0^*)
\,  \1(U_j\leq u_j)  \Pb_0(d\uu)\Big].
\label{formula_Sigma}
\end{align}
Alternatively, the asymptotic variance of $\hat\theta_n$ can be estimated by bootstrap resampling (see below).

\begin{Remark}
\label{GausianOrNot}
If the map $\ww\mapsto \nabla_\theta\ell(\ww;\theta_0^*)$ is ``sufficiently
regular'', the limiting law of $\sqrt{n}\big(\hat{\theta}_n - \theta_0^*\big)$ may still be Gaussian under Condition~\ref{cond_normality_4_ter} even if
\ref{cond_normality_4_bis} is not fulfilled. Indeed, this law is deduced from the weak convergence of
integrals as $\int \sqrt{n}(\hat{\Pb}_n-\Pb_0)(\ww)\, \nabla_\theta\ell(d\ww;\theta_0^*)$ (Equation~(\ref{Ibyparts}) in the proofs).
It is well-known that integration is a way of regularizing potentially discontinuous processes.
In particular, $\sqrt{n}\big(\hat{\theta}_n - \theta_0^*\big)$ is asymptotically normal
if $\ww\mapsto \nabla_\theta\ell(\ww;\theta_0^*)$ has an integrable density with respect to the
Lebesgue measure on $[0,1]^d$: apply Theorem 4.5 in~\cite{Segers2014} and the remark that follows, using the fact that the limiting copula process
has bounded trajectories in every case.
\end{Remark}

\begin{Remark}
\label{WeightedCopProcess}
Theorem~\ref{Th_AN} relies on the weak convergence of the usual empirical copula process and an integration by parts trick. If the map
$\ww\mapsto \nabla_\theta\ell(\ww;\theta_0^*)$ is not of bounded variation, as required in Condition~\ref{cond_normality_5_bis},
an alternative method would be to invoke the weak convergence of the
weighted empirical process (\cite{BerghausBV}, Theorem 2.2) in Equation~(\ref{Ibyparts}). This is relevant when $g(\ww)|\nabla_\theta\ell(d\ww;\theta_0^*)|$ defines a finite measure, by setting
$$ g(\ww)= \min \Big\{ \bigwedge_{k=1}^d u_k, \bigwedge_{k=1}^d \big( 1 - \min_{j\neq k} u_j  \big)\Big\}^{\omega},\; \omega \in (0,1/2).$$
The price to be paid for this strategy would be to require the existence and a certain amount of regularity for the second order derivatives of $C_0$, say
$$\partial^2 C_{0}(\uu)/\partial u_j\partial u_k \leq K\Big[ \max \big\{ u_j(1-u_j),u_k(1-u_k) \big\} \Big]^{-1},\; K>0,$$
for every $\uu\in V_j\cap V_k$ and every indices $j$ and $k$ in $\{1,\ldots,d\}$. Both ways of reasoning seem to be complementary, but without any clear hierarchy between them.
\end{Remark}

In canonical maximum likelihood estimation of semi-parametric models, the asymptotic normality of the copula parameter is usually obtained by similar techniques but using slightly
different assumptions: see, e.g.,~\cite{Genest1995,Chen2005, tsukahara2005}.
In such a situation, the loss function $\ell$ is the copula log-likelihood and Condition \ref{cond_normality_5_bis} should then hold on the score function rather than on $\nabla_{\theta}\ell(\, \cdot \, ;\theta_0^*)$. Unfortunately, the bounded variation assumption is violated by many popular copula families with unbounded copula score functions such as the Gaussian copula.
Hence, it is not possible to establish the asymptotic normality of CML-estimators
for the latter copula family using the same set of assumptions as in Theorem~\ref{Th_AN}.
Our MMD estimator is most often less demanding.
Indeed, its loss function is typically obtained by integrating copula densities, inducing a ``regularization procedure''.
In other words, conditions of regularity as Condition~\ref{cond_normality_5_bis} should be satisfied more easily in the MMD case compared to the usual CML method (even if this statement is not a universal rule).

\mds
Nonetheless, in every case, we can still rely on another set of technical assumptions, as for the CML method.
Now, we provide the following result adopting this alternative formulation, whose assumptions naturally hold for the Gaussian copula and can be checked by a direct analysis. 


\begin{Condition}
For any $\ww\in (0,1)^d$, $ \big\| \nabla_{\theta}\ell(\ww;\theta_0^*) \big\| \leq C_1 \prod_{k=1}^d \{w_k(1-w_k)\}^{-a_k} $
for some constants $C_1$ and $a_k\geq0$ such that
$$
\mathbb{E}\Big[\prod_{k=1}^d \{U_k(1-U_k)\}^{-2a_k}\Big]<+\infty .
$$
Moreover, for any $\ww\in(0,1)^d$ and any $k=1,\dots,d$,
$$
\big\| \nabla^2_{\theta,w_k}\ell(\ww;\theta_0^*) \big\| \leq C_2 \, \{w_k(1-w_k)\}^{-b_k} \prod_{j=1,j\ne k}^d \{w_j(1-w_j)\}^{-a_j},
$$
for some constants $C_2$ and $b_k>a_k$ such that
$$
\mathbb{E}\Big[\{U_k(1-U_k)\}^{\zeta_k-b_k} \prod_{j=1,j\ne k}^d \{U_j(1-U_j)\}^{-a_j}\Big]<+\infty,
$$
for some $\zeta_k\in(0,1/2)$.
\label{cond_normality_4}
\end{Condition}
Under the latter conditions, the partial derivatives of $\ell(\ww,\theta)$ are allowed to blow up at the boundaries of $[0,1]^d$, but not ``too quickly''.
Such conditions are well-known in the copula literature: see Assumption A.3 in~\cite{Chen2005} or Assumption A.1 in~\cite{tsukahara2005}.
Therefore, we get the same result as in Theorem~\ref{Th_AN}.
\begin{Theorem}
\label{Th_AN_bis}
If Conditions \ref{cond_consistency1}-\ref{cond_normality_plus} and \ref{cond_normality_4} 
are fulfilled, then the MMD estimator $\hat{\theta}_n$ is asymptotically normal:
$\sqrt{n}(\hat\theta_n - \theta_0^*) \xrightarrow[n\rightarrow+\infty]{\mathcal{L}} \mathcal{N}(0,B^{-1}\Sigma B^{-1}).$
\end{Theorem}

The beginning of the proof involves a first-order decomposition as in the proof of Theorem~\ref{Th_AN}. Nonetheless, instead of invoking integration by parts, it relies on some results about multivariate rank statistics that have been obtained by Ruymgaart and his co-authors in the 70's: see Proposition 2 in~\cite{Chen2005}.

\begin{proof}

As in the proof of Theorem~\ref{Th_AN}, we have under Conditions \ref{cond_consistency1} to \ref{cond_normality_3}:
$$ 0 = \nabla_\theta L_n(\theta_0^*) + H_n (\hat\theta_n-\theta_0^*),\;\;
\text{and} \; \;
H_n \xrightarrow[n\rightarrow+\infty]{\Pb_0 - a.s.} B .$$
Moreover, according to Lemma 2 in~\cite{Chen2005} applied to $J=\nabla_\theta\ell(\cdot;\theta_0^*)$ and $w_j(v)=\left(v(1-v)\right)^{\zeta_j}$, Condition \ref{cond_normality_4} directly leads to:
$$
\sqrt{n}\Big(\nabla_\theta L_n(\theta_0^*)-\mathbb{E}\big[\nabla_\theta\ell(\WW;\theta_0^*)\big] \Big) \xrightarrow[n\rightarrow+\infty]{\mathcal{L}} \mathcal{N}(0,\Sigma) ,\; \text{with}
\label{behavior_nabla_L_n}
$$
where $\Sigma$ is given in~(\ref{formula_Sigma}).
Condition \ref{cond_normality_plus} yields
$\sqrt{n}\nabla_\theta L_n(\theta_0^*) \xrightarrow[n\rightarrow+\infty]{\mathcal{L}} \mathcal{N}(0,\Sigma)$.
Finally, as previously, we obtain
$$
\sqrt{n}(\hat\theta_n - \theta_0) = H_n^{-1} \sqrt{n}\nabla_\theta L_n(\theta_0^*) \xrightarrow[n\rightarrow+\infty]{\mathcal{L}} \mathcal{N}(0,B^{-1}\Sigma B^{-1}) .
$$
\end{proof}

\mds

The limiting laws obtained in Theorem~\ref{Th_AN} and~\ref{Th_AN_bis} are most often complex, even in the case of Gaussian limit laws.
Once pseudo-observations are managed, particularly through empirical copula processes, it is common practice to rely on bootstrap schemes.

\mds


Any bootstrap scheme can be invoked as long as it is valid to evaluate the limiting law of the empirical copula process: see~(\ref{Ibyparts}) in our proofs.
Under Condition 9 (resp. Conditions 10-11), its weak convergence in $\ell^{\infty}([0,1]^d)$ (resp. $L^q([0,1]^d)$ for some $q>1$) is sufficient.
In the former case, we can rely on Efron's nonparametric bootstrap (\cite{JD2004}), the multiplier bootstrap in~\cite{RemillardScaillet}, among others.
In the latter case, apply another version of the multiplier bootstrap as defined in~\cite{Segers2014} (see the remark at the top of p. 1611).
And, in the case of a correctly specified copula model, the parametric bootstrap (\cite{GenestRemillard2008}) could surely be invoked too.

\mds

To be specific, the calculation of our nonparametric bootstrap estimator requires resampling every observation in the initial sample with replacement, yielding a bootstrap sample
$\Sc_n^*=\big(\XX_1^*,\ldots,\XX_n^*\big)$.
The associated empirical measure is
$$\Pb_n^*=n^{-1}\sum_{i=1}^n \delta_{\XX^*_i}=n^{-1}\sum_{i=1}^n W_{i,n} \delta_{\XX_i},    $$
where the vector of weights $(W_{1,1},\ldots,W_{n,n})$ is drawn following a $n$ multinomial law with success probabilities $(1/n,\ldots,1/n)$. We deduce
the bootstrapped empirical process as $\sqrt{n}\big(\hat\Pb_n^* - \hat\Pb_n \big)$, where $\hat\Pb_n^*$ denotes the empirical measure of the pseudo-sample obtained from $\Sc_n^*$. Exactly as for $\hat\theta_n$, one gets a bootstrapped estimator $\hat\theta^*_n$, but working on $\Sc_n^*$ instead of the initial sample.
The asymptotic laws of $\sqrt{n}(\hat\theta^*_n - \hat\theta_n)$ and $\sqrt{n}(\hat\theta_n - \theta_0)$ will
be the same because the limiting laws of $\sqrt{n}(\hat\Pb_n^*- \hat\Pb_n)$ and $\sqrt{n}\big( \hat\Pb_n-\Pb_0\big)$ are similar in~(\ref{Ibyparts}).

\mds
For the multiplier bootstrap (\cite{Segers2014}, Section 4.2), consider i.i.d. weights $(\xi_i)_{1\leq i \leq n}$, with both mean and variance equal to one. These weights
satisfy $\int \sqrt{\Pb(\xi_i>t)}\, dt <\infty$ and are independent of the sample.
Introduce the cdf $G_n^*(\xx)=n^{-1}\sum_{i=1}^n \xi_i \1(\XX_i\leq \xx)$ on $\Rb^d$ and its margins $G_{n,k}^*$, $k\in\{1,\ldots,d\}$. Build the pseudo-sample
$(\VV_i^*)_{i=1,\ldots,n}$ where $V^*_{i,k}=G_{n,k}^*(X_{i,k})$ for any $k$, the associated empirical copula $\tilde C_n^*$ and the empirical copula process $\sqrt{n}\big(
\tilde C_n^* - C_n\big)$.
The bootstrapped estimator of $\theta_0$ is then obtained by MMD minimization, but replacing the initial pseudo-sample $(\hat\UU_i)_{i=1,\ldots,n}$ by
$(\VV^*_i)_{i=1,\ldots,n}$, and the same arguments as above apply.

\mds

Recently, subsampling has been proposed as an interesting alternative to bootstrap estimates of functionals of many empirical copula processes, possibly smoothed or weighted (\cite{Kojadinovic2019}). This technique is valid when our Condition~\ref{cond_normality_4_bis} is satisfied and when the usual empirical process of $(\UU_i)_{i=1,\ldots,n}$ is weakly convergent in $\ell^{\infty}([0,1]^d)$ to a tight centered Gaussian process. In particular, the latter result applies when our $\XX$-sample is a stretch from a strongly mixing stationary sequence.

\subsection{Examples}

Now, let us check that the previous asymptotic results can be applied for two usual bivariate copula families, here the Gaussian and the Marshall-Olkin copulas.
In this subsection, when we assume that the model is well-specified, i.e. that the law of the observations belongs to the considered parametric family,
the pseudo-true parameter $\theta_0^*$ is simply the true underlying parameter and is denoted by $\theta_0$.

\mds
In both cases, we will use some characteristic Gaussian-type kernel $K_U$ defined as
\begin{equation}
 K_h(\uu,\vv)=\exp\Big\{ -\frac{( h(u_1) - h(v_1))^2+( h(u_2) - h(v_2))^2  }{\gamma^2}  \Big\} ,
\label{def_K_U}
 \end{equation}
for some injective map $h:[0,1]\mapsto \Rb$ and some tuning parameter $\gamma>0$ (see, e.g.,~\cite{Christmann}, Th. 2.2).
Indeed, the latter function $K_h$ is a kernel: let $\zeta:\Rb^2 \rightarrow \Fc$ be the feature map that is associated with the usual Gaussian kernel $K_G$, i.e.
$ K_G(\xx,\yy)= \langle \, \zeta(\xx) , \zeta(\yy) \, \rangle_{\Fc}$, where the Gaussian kernel is defined for $\xx,\yy\in\mathbb{R}^2$ by
$$
K_G(\xx,\yy)=\exp\Big\{ -\frac{( x_1 - y_1)^2+( x_2 - y_2)^2  }{\gamma^2}  \Big\} .
$$
Then, the feature map that defines $K_h$ is simply $\psi : [0,1]^2\rightarrow \Fc$ given by
$ \psi(\uu)=\zeta\big( h(u_1),h(u_2)\big)$ for every $\uu\in (0,1)^2$, and $K_h$ inherits from $K_G$ its ``characteristic'' property.

\mds

Hereafter, we shall denote by $\Phi$ and $\phi$ the cumulative distribution function and the probability density function of the standard normal distribution, respectively.
Then, a natural choice is to set $h(u)=\Phi^{-1}(u)$.
The latter kernel will simply be denoted by $K_U$.
Even if it is possible to always choose the usual Gaussian kernel $K_G$ by setting $h(u)=u$, we have observed that $K_U$ provides better numerical results in some situations. We refer the reader to the simulation study for a detailed comparison.
Moreover, it is sometimes simpler to use $K_U$ rather than $K_G$.
For example, in the case of Gaussian copulas, the criterion $L_0$ can be analytically calculated when $K=K_U$ (see Appendix F
), contrary to $K=K_G$.
Note that it is not so surprising that $K_U$ provides better empirical results than $K_G$.
Indeed, it is a common procedure in copula modeling to push back the sample observations on $\mathbb{R}^d$ using Gaussian quantile functions componentwise.
This trick spreads the data cloud and often improves inference.
At the opposite, our conditions of regularity for Marshall-Olkin copulas can be checked only when the kernel is $K_G$.

\subsubsection{Gaussian copulas}

Let us consider two-dimensional Gaussian copulas $C_\theta(\uu)=\Phi_{2,\theta}\big( \Phi^{-1}(u_1), \Phi^{-1}(u_2) \big)$, indexed by $\theta\in
(-1,1)$.
Here, $\Phi_{2,\theta}$ denotes the cdf of a bivariate Gaussian centered vector $(X_1,X_2)$, $\Eb[X_k^2]=1$, $k=1,2$, and $\Eb[X_1 X_2]=\theta$.
The associated copula density has been given in Equation~(\ref{def_cop_gauss}).


\begin{Proposition}
\label{AN_Gaussian_cop}
Assume that the true underlying copula is $C_{\theta_0}$ for some parameter $\theta_0\in (-1,1)$.
Then, when $K\in \{K_U,K_G\}$ and $\gamma^2<2$, the estimator $\hat\theta_n$ given by~(\ref{RKHS_criterion}) is strongly consistent and $\sqrt{n}(\hat\theta_n -\theta_0)$
is asymptotically normal.
\end{Proposition}
The proof is deferred to Appendix C. 
For the sake of illustration, we will verify the conditions of validity of Theorem~\ref{Th_AN}, even if those of Theorem~\ref{Th_AN_bis} can be checked too.
In this proof, it is stated that the term $B$ that appears in the asymptotic variance of $\hat \theta_0$ when $K=K_U$ has the closed-form expression
$$ B_G(\theta_0) = \frac{3 \gamma^2\big\{ (2+\gamma^2/2)^2 + 8\theta_0^2 \big\}}{2\{(2+\gamma^2/2)^2 - 4 \theta_0^2\}^{5/2}}>0\cdot$$

\mds

Now, let us deal with the general case of misspecification.
\begin{Corollary}
\label{AN_Gaussian_cop_cor}
Assume that the true underlying copula is $C_0$ and $K\in \{K_U,K_G\}$ with $\gamma^2<2$.
If the estimator $\hat\theta_n$ given by~(\ref{RKHS_criterion}) is strongly consistent to $\theta_0^*\in (-1,1)$ that satisfies the first-order Condition~\ref{cond_normality_plus} and if $B>0$, then
$\sqrt{n}(\hat\theta_n -\theta^*_0)$ is asymptotically normal.
\end{Corollary} The proof is given in Appendix D. 
When a Gaussian copula is contaminated by a fixed bivariate copula $\bar C$, then $C_0=(1-\eps)C_{\theta_0}+\eps\bar C$, and
the real number $B$ is now
$$B=\int \nabla_{\theta,\theta}^2  \ell(\uu;\theta_0^*) C_0(d\uu) = (1-\eps)B_G(\theta^*_0) + \eps \int \nabla_{\theta,\theta}^2 \ell(\uu;\theta_0^*) \bar C(d\uu).$$
Here, we have assumed the consistency of $\hat\theta_n$ because we cannot exclude the existence of several minimizers of $L_0$ in general, even if it is a very unlikely situation.

\subsubsection{Marshall-Olkin copulas}

By definition (\cite{Nelsen2007}, Section 3.1.1), the bivariate Marshall-Olkin copula is defined on $[0,1]^2$ as
\begin{equation}
 C_\theta(u,v)= u^{1-\alpha} v \1(u^\alpha \geq v^\beta) + u v^{1-\beta} \1(u^\alpha < v^\beta),
 \label{def_MO}
 \end{equation}
for some parameter $\theta=(\alpha,\beta)$, $0<\alpha,\beta <1$.
This copula has no density with respect to the Lebesgue measure on the whole $[0,1]^2$.
The absolutely continuous part of $C_\theta$ (with respect to the Lebesgue measure) is defined on $[0,1]^2\setminus \Cf$, where $\Cf=\{(u,v)\in [0,1]^2 \setminus u^\alpha = v^\beta\}$.
The singular component is concentrated on the curve $\Cf$, and $\Pb( U^\alpha= V^\beta)=\alpha\beta/(\alpha +\beta-\alpha\beta)=:\kappa$, when $(U,V)\sim C_\theta$.
With the same notation as in~\cite{Nelsen2007}, $C_\theta(u,v)=A_\theta(u,v)+S_\theta(u,v)$, where, for every $(u,v)\in [0,1]^2$,
$S_\theta(u,v)=
\kappa\big\{ \min(u^\alpha,v^\beta) \big\}^{1/\kappa}$ and
\begin{align*} A_\theta(u,v)
& =\int_0^u\int_0^v \frac{\partial^2 C_\theta}{\partial u\partial v}(s,t)\,ds\,dt
\\
& = \int_0^u\int_0^v \big\{ (1-\alpha)s^{-\alpha}\1(s^\alpha > t^\beta)+ (1-\beta)t^{-\beta}\1(s^\alpha < t^\beta) \big\} \,ds\,dt.
\end{align*}

\mds
Let us calculate $\Eb[\psi\big( U,V\big)]$, $(U,V)\sim C_\theta$, for any measurable map $\psi$, to be able to calculate $\ell(\ww,\theta)$ for our bivariate Marshall-Olkin model.
Given a small positive real number $\delta$, let us first evaluate the mass along $\Cf$, when the abscissa and the ordinate belong to $[u,u+\delta]$ and $[v,v+\delta]$ respectively: if $u^\alpha=v^\beta$ and $\delta \ll 1$,
\begin{eqnarray*}
\lefteqn{S_\theta(u+\delta,v+\delta)-S_\theta(u+\delta,v)-S_\theta(u,v+\delta)+S_\theta(u,v) }\\
&=&\kappa \min\big\{ (u+\delta)^\alpha,(v+\delta)^\beta \big\}^{1/\kappa} - \kappa u^{\alpha/\kappa}     \\
&\simeq & \delta \alpha u^{\alpha/\kappa - 1} \1(\alpha v \leq \beta u)+  \delta \beta v^{\beta/\kappa - 1} \1(\alpha v > \beta u) \\
&\simeq & \delta \alpha u^{\alpha/\beta - \alpha} \1(\alpha v \leq \beta u)+  \delta \beta u^{1-\alpha} \1(\alpha v > \beta u),
\end{eqnarray*}
providing the density along the curve $\Cf$.
Therefore, we obtain
\begin{equation}
\Eb[\psi\big( U,V\big)] = \int \psi(s,t) \frac{\partial^2 C_\theta}{\partial u\partial v}(s,t)\,ds\,dt + \int \psi(u,v)\, S_\theta(du,dv)=:I_1+I_2,
\label{psi_UV}
\end{equation}
\begin{equation}
 I_1 =  \int \psi(s,t) \big\{ (1-\alpha)s^{-\alpha}\1(s^\alpha > t^\beta)+ (1-\beta)t^{-\beta}\1(s^\alpha < t^\beta) \big\} \,ds\,dt.
 \label{def_I1}
 \end{equation}

Let $(\bar u_{\alpha,\beta},\bar v_{\alpha,\beta})$
be a point of $\Cf$ such that $\alpha\bar v_{\alpha,\beta}=\beta \bar u_{\alpha,\beta}$.
It is easy to check that such a point exists in $[0,1]^2$ and is unique, except when $\alpha=\beta$.
In the latter case, the couple $(\bar u_{\alpha,\beta},\bar v_{\alpha,\beta})$ may be arbitrarily chosen along the main diagonal of $[0,1]^2$.
Then, we get
\begin{equation}
 I_2=\int \psi(u,v)\, S_\theta(du,dv)=
\int_0^{\bar u_{\alpha,\beta}} \psi(u,u^{\alpha/\beta})\,\beta u^{1-\alpha}\, du +
\int_{\bar u_{\alpha,\beta}}^{1} \psi(u,u^{\alpha/\beta})\,\alpha u^{\alpha/\beta-\alpha}\, du ,
 \label{def_I2}
  \end{equation}
with $\bar u_{\alpha,\beta}=\big( \beta/\alpha \big)^{\beta/(\alpha-\beta)}$ when $\alpha\neq \beta$ and $\bar u_{\alpha,\alpha}=e^{-1}$.
The latter value has been chosen so that the map
$(\alpha ,\beta) \mapsto \bar u_{\alpha,\beta}$ is continuous on the whole set $(0,1)^2$, i.e. even at the main diagonal.
For most regular functions $\psi$, the latter integrals $I_1$, $I_2$ and then $\Eb[\psi\big( U,V\big)]$ are continuous functions of $(\alpha,\beta)$.

\begin{Proposition}
\label{AN_MO}
For almost any true parameter $\theta_0=(\alpha_0,\beta_0)$ that belongs to the interior of $\Theta=[\epsilon,1-\epsilon]^2$ for some $\epsilon\in(0,1/2)$, the estimator $\hat\theta_n$
given by~(\ref{RKHS_criterion}) is strongly consistent, using the kernel $K_U$ or $K_G$.
Moreover, when $K=K_G$ and $B$ is positive definite, $\sqrt{n}(\hat\theta_n -\theta_0)$ is weakly convergent.
\end{Proposition}
See the proof in Appendix E. 
When the latter limiting law of $\sqrt{n}(\hat\theta_n - \theta_0)$ exists, it is not Gaussian in general. It could be numerically evaluated by usual resampling techniques,
as the consistent bootstrap scheme in~\cite[Section 4.2]{Segers2014}.

\medskip

\begin{Remark}
The difficulty to state the limiting law of $\sqrt{n}(\hat\theta_n - \theta_0)$ with $K=K_U$ arises from the second-order derivatives of
$\ell(\ww,\theta)$ with respect to $\theta$. To be short, at some stage, one has to deal with integrals as
$$ \int \exp\big\{-\frac{( x - y)^2}{\gamma^2}-\frac{( x_{\alpha/\beta} - y_{\alpha/\beta} )^2}{\gamma^2}\big\} \frac{x^a y^{\bar a}\Phi(x)^b \Phi(y)^{\bar b}}{\phi^2(x_{\alpha/\beta})} \ln^c \Phi(x)\ln^{\bar c} \Phi(y) \phi(x)\phi(y)\, dx\, dy $$
by setting $t_\nu=\Phi^{-1}\big( \Phi(t)^\nu  \big)$ for any $\nu \geq 0$ and any real number $t$.
Here, $(a,b,c,\bar a,\bar b,\bar c)$ denotes a vector of nonnegative real numbers.
It can be proved that the latter integral is not convergent, even in the simplest case $\alpha=\beta$ and all the other constants are zero.
\end{Remark}

In the case of general misspecification, a similar result is still valid.
\begin{Corollary}
\label{AN_MO_cop_cor}
Assume that the true underlying copula $C_0$ is arbitrary.
If the estimator $\hat\theta_n$ given by~(\ref{RKHS_criterion}) is strongly consistent to $\theta_0^*\in (-1,1)$ that satisfies the first-order Condition~\ref{cond_normality_plus}, then the same results as in Proposition~\ref{AN_MO} apply, replacing $\theta_0$ by $\theta_0^*$.
\end{Corollary}
The arguments of the proof are exactly the same as in Appendix E. 

\section{Implementation and experimental study}
\label{section:experiments}

In this section, we compare the MMD estimator to the CML and the moment estimator on simulated data. The CML and the method of moments by inversion of Kendall's tau are implemented in the R package VineCopula~\cite{VineCopulaR}. We implemented the MMD estimator using the stochastic gradient algorithm described in~\cite{Cherief2019}. This procedure requires sampling from the copula model we want to estimate. For this, we used again VineCopula. Note that our implementation of the MMD estimator is itself available as the R package MMDCopula~\cite{MMDCopula}.

\subsection{Implementation via stochastic gradient and the MMDCopula package}

We start by a short description of the algorithm implemented in our R package~\cite{MMDCopula} to compute the MMD estimator in the bivariate case. It is of course possible to use the vine-copula procedure to decompose higher-dimensional copulas into bivariate ones. The main idea is differentiating the criterion~\eqref{RKHS_criterion:density}. Under suitable regularity assumptions on the copula density $c_\theta$ with respect to the Lebesgue measure on $\Uc$, we have
\begin{align*}
\frac{d}{d\theta} & \Big[
    \int K_U(\uu,\vv) c_\theta (\uu) c_\theta (\vv) \, d\uu \, d\vv
- \frac{2}{n}\sum_{i=1}^n \int K_U(\uu,\hat \UU_i) c_\theta (\uu) \, d\uu \Big]
\\
& = 2 \int K_U(\uu,\vv) \frac{d \ln c_\theta (\uu)}{d\theta} c_\theta (\uu) c_\theta (\vv) \, d\uu \, d\vv
- \frac{2}{n}\sum_{i=1}^n \int K_U(\uu,\hat \UU_i)  \frac{d \ln c_\theta (\uu)}{d\theta} c_\theta (\uu) \, d\uu
\\
& = 2 \, \mathbb{E}\Big[
\frac{d \ln c_\theta (\UU)}{d\theta} \big\{ K_U(\UU,\VV)   - \frac{1}{n}\sum_{i=1}^n K_U(\UU,\hat{\UU}_i) \big\}
\Big],
\end{align*}
where the expectation is taken with respect to $\UU$ and $\VV$, that are independently drawn from $C_\theta$ (a formal statement can be found in~\cite{Cherief2019}). Even though this expectation is usually not available in closed form, it is possible to estimate it by Monte Carlo to use a stochastic gradient descent. That is, we fix a starting point, a step size sequence $(\eta_t)_{t\geq 0}$, and iterate:
\begin{equation*}
\left\{
 \begin{array}{l}
  \text{draw } \UU_1^\star,\dots,\UU_n^\star,\VV_1^\star,\dots,\VV_n^\star \sim C_{\theta_t} \text{ i.i.d},
  \\
  \theta_{t+1} \leftarrow  \theta_t - 2\eta_t n^{-2} \sum_{i,j=1}^n
\frac{d \ln c_\theta (\UU_j^\star)}{d\theta}_{|\theta=\theta_t} \big\{ K(\UU_j^\star,\VV_i^\star) - K_U(\UU_j^\star,\hat{\UU}_i) \big\} .
 \end{array}
 \right.
\end{equation*}
In practice, we take $\eta_t=1/\sqrt{t}$ as recommended in~\cite{Cherief2019}. We perform 200 iterations, and return the average of $\theta_t$ over the last $100$ iterations.

\medskip

The implementation of this algorithm requires (i) to be able to sample from $C_\theta$ and (ii) to compute $c_\theta$ and its partial derivative with respect to $\theta$. A list of copula densities and their differentials can be found in~\cite{Sch14} and is implemented in VineCopula~\cite{VineCopulaR}. Some procedures to sample from $C_\theta$ can also be found in VineCopula. The same ideas can be adapted even if the latter copula density does not exist on the whole hypercube, as for the Marshall-Olkin copula. In the latter case with $\alpha=\beta$, we implemented our own sampler and considered the copula density with respect to the measure given by the sum of the Lebesgue measure on $[0,1]^2$ plus the Lebesgue measure on the first diagonal.

In theory, the criterion in~\eqref{RKHS_criterion} has no reason to be convex in $\theta$. Therefore, it is possible that the algorithm gets stuck in a local minimum. In order to avoid this situation, we propose two possible strategies: 1) starting from a random initialization and 2) starting from the empirical Kendall's tau and the associated $\theta$ values. We compared these two strategies in a set of experiments in the supplementary material. In the non contaminated case, the Kendall's tau initialization is slightly better (especially for small $\gamma$'s) but both strategies are comparable. However in a contaminated case, the random initialization becomes better. We suspect Kendall's tau might be close to a local minimizer of the MMD in the latter case. In our package and in our simulations, the random initialization is the default mode. The convergence of stochastic gradient algorithms for MMD mininimization in a general framework is discussed in~\cite{Cherief2019}.

\medskip

Also, note that it is possible to use a quasi Monte Carlo rather than a Monte Carlo sampling scheme. In our package MMDCopula~\cite{MMDCopula}, we give the user the possibility to choose the sampling scheme for the $\UU_j$'s and the $\VV_i$'s separately. In all our simulations, we observed that the use of Monte Carlo on the $\UU_j$
and of quasi Monte Carlo on the $\VV_i$'s led to the best results, so this setting is chosen by default in our package, and it was also used in the following experiments. An important point is that the gradient method is {\it not} invariant by reparametrization. In order to deal with gradient descents in compact sets only, we decided to parametrize all the copulas by their Kendall's tau (apart from the Marshall-Olkin copula, implemented in the case $\alpha=\beta$, that is parametrized by $\alpha$ and does not use quasi Monte Carlo).

\medskip

Finally, in the MMDCopula package, the estimator $\hat\theta_n$ can be computed for five different kernels. In the following simulations, we worked with
the Gaussian kernel $k_U(\UU,\VV) = \exp(- \|h(\UU) - h(\VV) \|_2^2 / \gamma^2 ) $, the exp-$L_2$ kernel $k_U(\UU,\VV) = \exp(- \|h(\UU) - h(\VV) \|_2 / \gamma ) $ and the exp-$L_1$ kernel $k_U(\UU,\VV) = \exp(- \|h(\UU) - h(\VV) \|_1 / \gamma ) $, where $h$ is either the identity or $\Phi^{-1}$ and is applied coordinatewise. A major question is then: how to calibrate $\gamma$, and which kernel to choose? We performed some experiments on synthetic data to answer this question. In Figure~\ref{FigureGamma}, we provide the MSE of the estimators based on these three kernels as a function of $\gamma$. A more complete study of the dependence of the MSE with respect to $\gamma$ in various models is provided in Appendix I.

\medskip

\begin{figure}[h]
 \begin{center}
 \includegraphics[height=7.5cm]{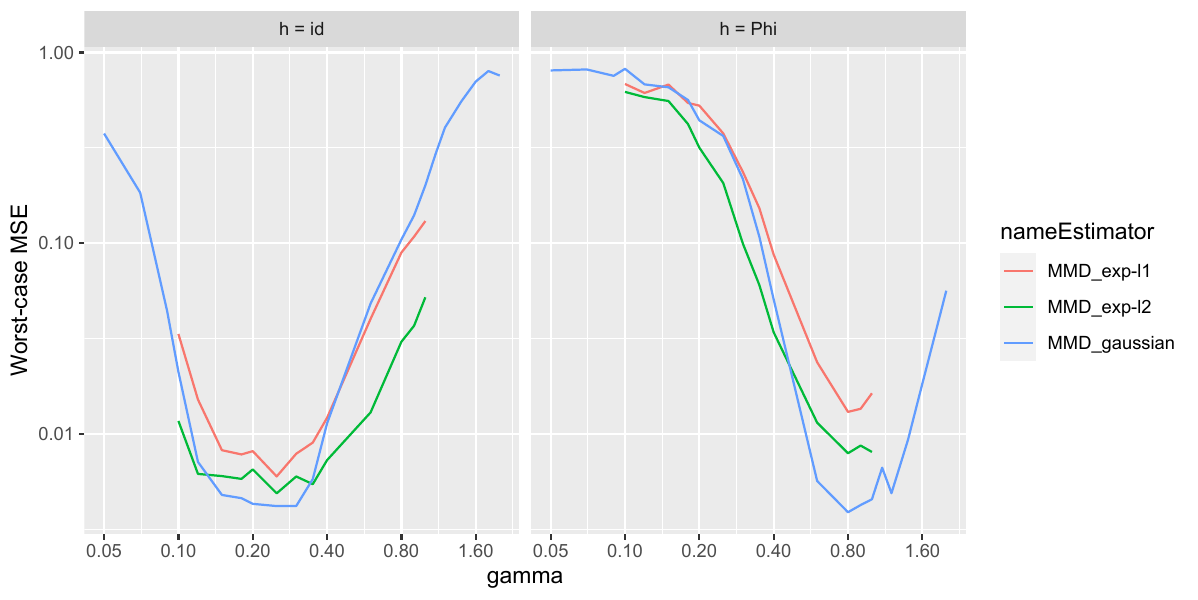}
 \caption{MSE of $\hat\theta_n$ based on the Gaussian kernel $k_U(\UU,\VV) = \exp(- \|h(\UU) - h(\VV) \|_2^2 / \gamma^2 ) $, the exp-$L_2$ kernel $k_U(\UU,\VV) = \exp(- \|h(\UU) - h(\VV) \|_2 / \gamma ) $ and the exp-$L_1$ kernel $k_U(\UU,\VV) = \exp(- \|h(\UU) - h(\VV) \|_1 / \gamma ) $, as functions of $\gamma$.}
 \label{FigureGamma}
 \end{center}
\end{figure}

In these experiments, $n=1000$ observations were sampled from the Gaussian copula, and the objective was to estimate the parameter of this copula. Each experiment was repeated $200$ times. Except in some experiments in the supplement used to calibrate $\gamma$, the true Kendall's tau was fixed as $\tau=0.5$.

\medskip

The take-home message is that, as far as the Gaussian copula is concerned and $n=1000$, the Gaussian kernel is the best one, whatever the choice of $h$. When $h$ is the identity map, the optimal $\gamma$ is $\gamma \simeq 0.25$. For $h(\uu)=\big(\Phi^{-1}(u_1),\Phi^{-1}(u_2)\big)$, the optimal value is $\gamma=0.80$. We performed similar experiments for other copula families. The results can be found in Appendix I. 
The optimal values for each family are set as default values in our package, and used in the following experiments.

Finally, note that we also discuss the computational cost in Appendix G (for n=1000, a MMD estimation takes around 4-7 seconds for most copula families).

\subsection{Comparison to CML on synthetic data}

We now compare the MMD estimators based on the Gaussian kernel (with two choices of $h$) to the canonical maximum likelihood (CML) estimator and the estimator based on the inversion of Kendall's tau (``Itau''). We would like to illustrate convergence when the sample size $n\rightarrow \infty$ and robustness to the presence of various type of outliers. We designed nine types of outliers.
\begin{itemize}
 \item {\it Uniform}: the outliers are drawn i.i.d from the uniform distribution $\mathcal{U}([0,1]^2)$.
 \item {\it Top-left}: the outliers belong to the top-left corner of $[0,1]^2$, that is, they are drawn i.i.d from $\mathcal{U}([0,q] \times [1-q,q])$ where $q = 0.001$.
 \item {\it Bottom-left}: the outliers belong to the bottom-left corner, that is, they are drawn i.i.d from $\mathcal{U}([0,q]^2)$.
 \item {\it Diagonal}: the outliers are uniform on the first diagonal.
 \item {\it Gauss 0.2}: the outliers are drawn from the Gaussian copula with a Kendall's tau equal to $0.2$.
 \item {\it Gauss -0.8}: the outliers are drawn from the Gaussian copula with a Kendall's tau equal to $-0.8$.
 \item {\it Frank -0.8}: the outliers are drawn from the Frank copula with a Kendall's tau equal to $-0.8$.
 \item {\it Clayton 0.5}: the outliers are drawn from the Clayton copula with a Kendall's tau equal to $0.5$.
 \item {\it Student 0.5 3df}: the outliers are drawn from the Student copula with a Kendall's tau equal to $0.5$ and $3$ degrees of freedom.
\end{itemize}
In each case, the data are sampled on $[0,1]^2$ from the desired copula.
Finally, the contaminated observations are rescaled by their rank in order to keep pseudo-uniform margins.

\mds
In a first series of experiments, we use the various estimators to estimate the parameter of the Gaussian copula. We compare their robustness to the presence of a proportion $\varepsilon$ of each type of outliers, when $\varepsilon$ ranges from $0$ to $0.05$. In a second time, we go beyond the Gaussian model: we replicate these experiments for the Frank copula, the Clayton copula, the Gumbel copula and the Marshall-Olkin copula. The results being quite similar, we save space by reporting only them for {\it top-left} outliers. In the last series of experiments, we come back to the Gaussian case, and illustrate the asymptotic theory. In this last experiment, we study the convergence of the estimators when $n$ grows in two situations: no outliers, or a proportion $\varepsilon \in \{0.05,0.1\}$ of {\it top-left} outliers.

\subsubsection{Robustness to various types of outliers in the Gaussian copula model}

For each type of outliers, and for each $\varepsilon$ in a grid that ranges from $0$ to $0.05$, we repeat $1000$ times the following experiment: the data are i.i.d from the Gaussian copula, the sample size is $n=1000$ and the parameter is calibrated so that $\tau=0.5$. Then, an exact proportion $\varepsilon$ of the data is replaced by outliers. We report the mean MSE of each estimator in Figure~\ref{FigureGaussOutliers}.

\medskip

\begin{figure}[p]
 \begin{center}
 \includegraphics[width=\textwidth]{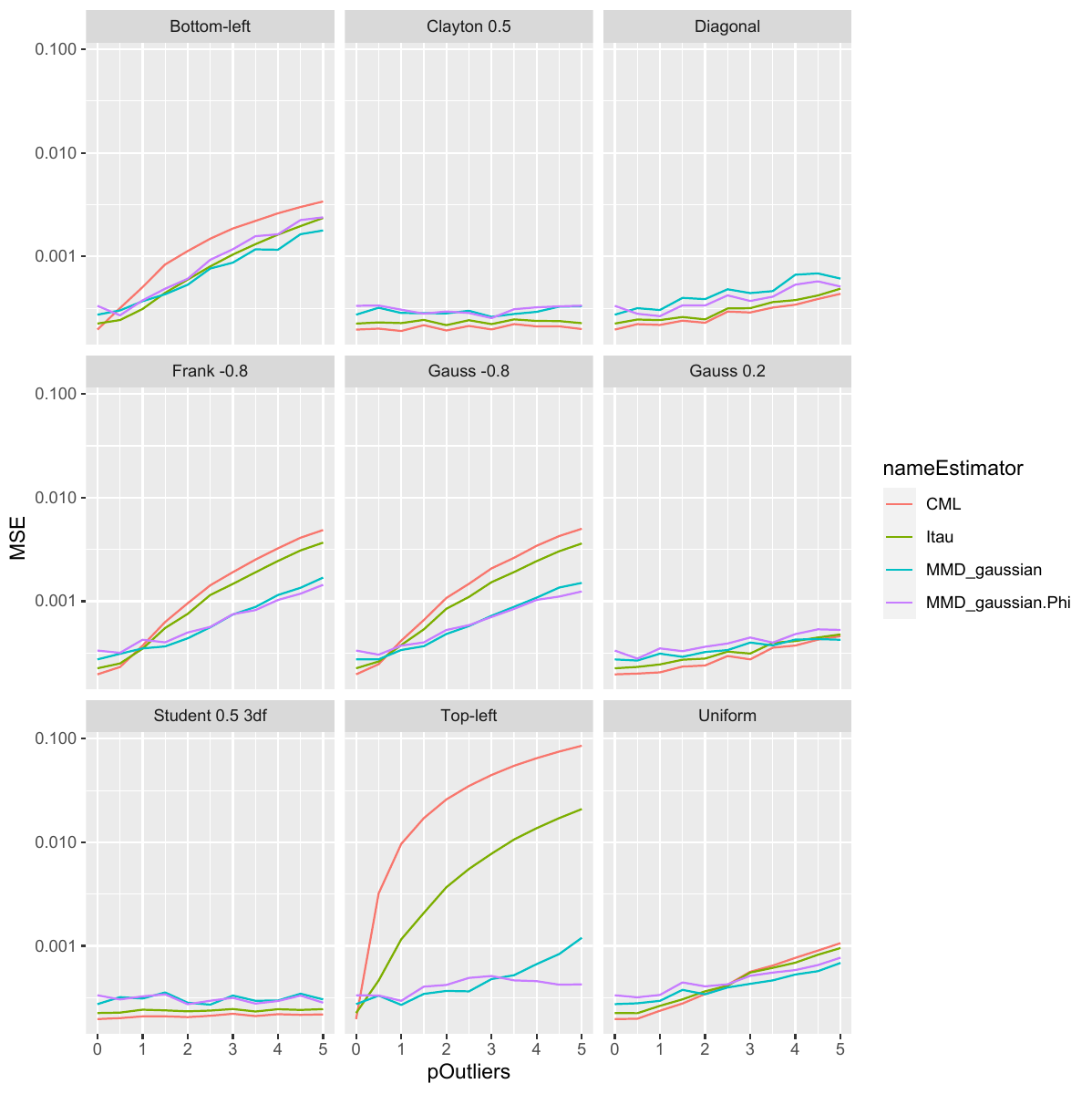}
 \caption{MSE of the MMD estimator with Gaussian kernel and $h(u)=u$, the MMD estimator with Gaussian kernel and $h(u)=\Phi^{-1}(u)$, the CML estimator and the method of moment based on Kendall's $\tau$, as a function of the proportion $\varepsilon$ of outliers. Sample size: $n=1000$, model: Gaussian copula. The title of each box gives the distribution of the contamination.}
 \label{FigureGaussOutliers}
 \end{center}
\end{figure}

When there are no outliers, CML yields the best estimator.
However, as soon as there is more than 2 or 3 percent of outliers, the MMD estimators become much more reliable when contamination arises from a distribution that
significantly differs from the reference Gaussian copula. Interestingly, the one based on $h(u)=u$ becomes equivalent to the one based on $h(u)=\Phi^{-1}(u)$ with {\it uniform} outliers, in terms of MSE.

\subsubsection{Robustness in various models}

Here, we replicate the previous experiments with other models: Clayton, Gumbel, Frank and Marshall-Olkin. In each case, the parameter was chosen so that $\tau=0.5$. We report the results in the case of {\it top-left} outliers in Figure~\ref{FigureOtherFamilies}.

\medskip

\begin{figure}[h]
 \begin{center}
 \includegraphics[width=15cm]{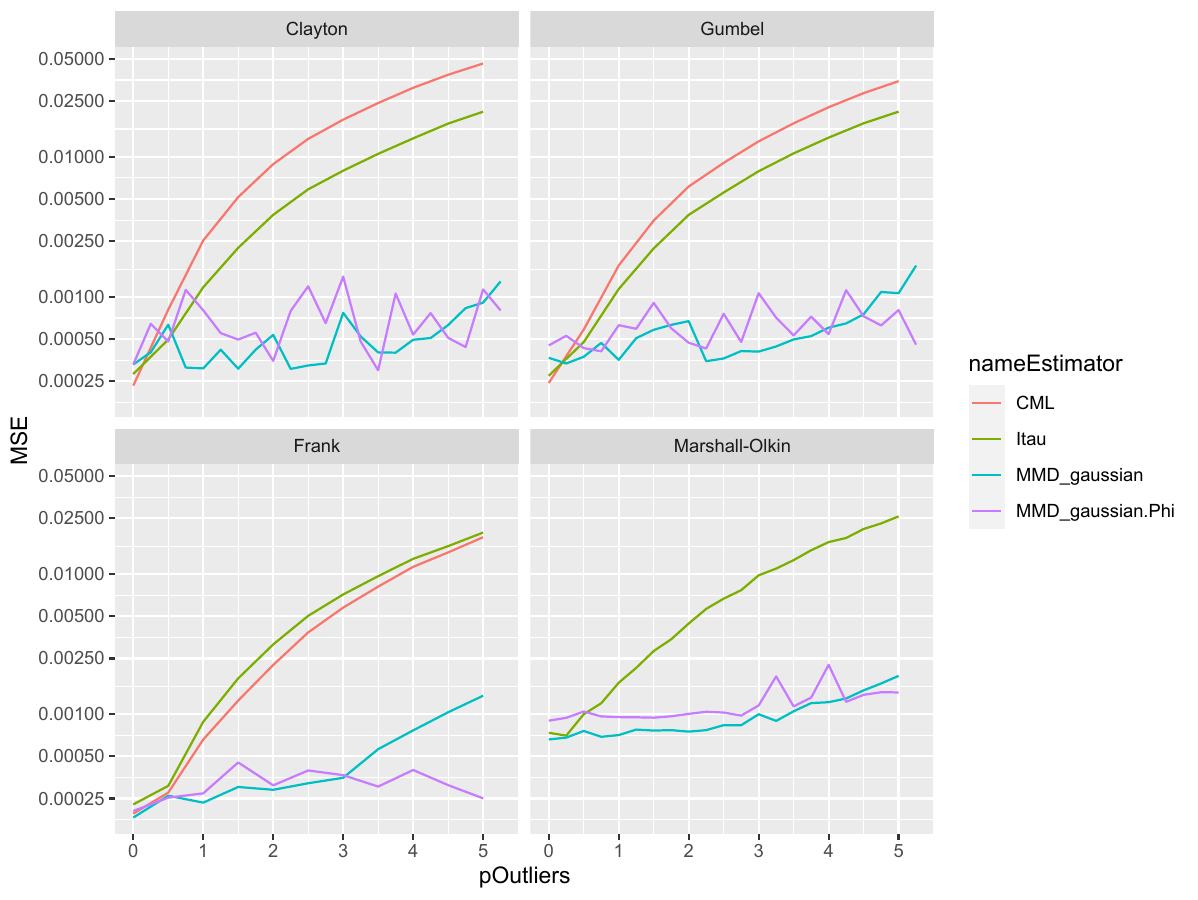}
 \caption{MSE of the MMD estimator with Gaussian kernel and $h(u)=u$, the MMD estimator with Gaussian kernel and $h(u)=\Phi^{-1}(u)$, the CML estimator and the method of moment based on Kendall's $\tau$, as a function of the proportion $\varepsilon$ of {\it top-left} outliers. Sample size: $n=1000$. Top-left: Clayton copula. Top-right: Gumbel copula. Bottom-left: Frank copula. Bottom-right: Marshall-Olkin copula.}
 \label{FigureOtherFamilies}
 \end{center}
\end{figure}

The conclusion remains unchanged: in all models, the MMD estimators are far more robust than the CML and the method of moments estimators.

\subsubsection{Convergence}

We finally come back to the Gaussian copula case. This time, we study the influence of the sample size $n$, ranging from $n=100$ to $n=5000$. We report the results of simulations without outliers ($\varepsilon=0.00$) and with {\it top-left} outliers ($\varepsilon=0.05$ and $\varepsilon=0.1$, independently of the sample size) in Figure~\ref{FigureGaussn}.

\medskip

\begin{figure}[h]
 \begin{center}
 \includegraphics[width=15cm]{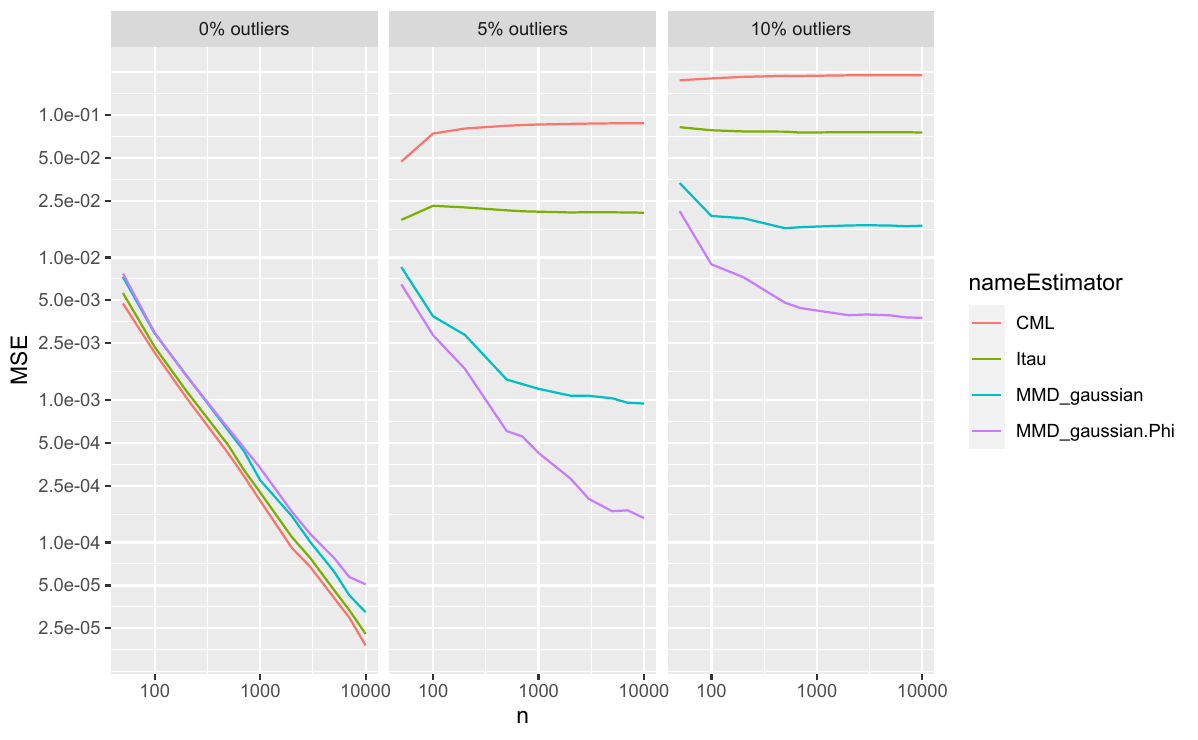}
 \caption{MSE of the MMD estimator with Gaussian kernel and $h(u)=u$, the MMD estimator with Gaussian kernel and $h(u)=\Phi^{-1}(u)$, the CML estimator and the method of moment based on Kendall's $\tau$, as a function of the sample size $n$. Model: Gaussian copula. Left: no outliers. Middle: a proportion $\varepsilon=0.05$ of outliers. Right: a proportion $\varepsilon=0.1$ of outliers.}
 \label{FigureGaussn}
 \end{center}
\end{figure}

When there are no outliers, we observe the $\sqrt{n}$ consistency of all the estimators, as predicted by the theory. The CML method yields the best estimator in this case. However, when there are outliers, the situation is dramatically different. All the estimators have an incompressible bias, and only their variances will decrease to $0$. However, we already observed that the MMD estimators are a lot more robust to outliers: indeed, here, their bias is (much) smaller than the other competing methods. Note that the hierarchy between the different methods is unaffected by the sample size.

\section{Conclusion}

We have shown that the estimation of semiparametric copula models by MMD methods yields consistent, weakly convergent and robust estimators.
In particular, when some outliers contaminate an assumed parametric underlying copula, the comparative advantages of our MMD estimator become patent.

\mds

To go further, many open questions would be of interest.
For instance, extending our theory to manage time series should be feasible.
Indeed, the theory of the weak convergence of empirical copula processes for dependent data
has been established in the literature; see, e.g.,~\cite{BuecherVolgushev}.
Moreover, finding a formal data-driven way of choosing the kernel tuning-parameter $\gamma$ would be useful.
Finally, in the case of highly parameterized models -- such as hierarchical Archimedean models (HAC), vines, or reliability models based on Marshall-Olkin copulas also called ``fatal shock'' models --, it could be interesting to introduce a penalization on $\theta$, for example as
$$
\tilde \theta_n \in  \arg\min_{\theta\in\Theta}
\int K_U(\uu,\vv) \Pb^U_\theta (d\uu) \, \Pb^U_\theta (d\vv) - \frac{2}{n}\sum_{i=1}^n \int K_U(\uu,\hat \UU_i) \Pb^U_\theta (d\uu) + \lambda \|\theta \|_1.$$
This idea would be different from the so-called ``regularized MMD'' in~\cite{Danafar2013} that is reduced to multiplying the first term on the right-hand side of the latter equation by a scaling factor.
To the best of our knowledge, the asymptotic or finite distance theory for the penalized MMD estimator $\tilde\theta_n$ still does not exist.
An interesting avenue for future research would be to fill this theoretical gap and to adapt this framework to copulas.

\mds

\noindent
{\bf Acknowledgements}
\noindent

\medskip
First, we thank both anonymous Referees for their insightful comments that led to many improvements of the paper.
Badr-Eddine Ch\'erief-Abdellatif acknowledges support of the UK Defence Science and Technology Laboratory (DSTL) and EPSRC under grant EP/R013616/1. This is part of the collaboration between US DOD, UK MOD and UK EPSRC under the Multidisciplinary University Research Initiative. Jean-David Fermanian has been supported by the labex Ecodec (reference project ANR-11-LABEX-0047).


\newpage

\appendix

\section{Refinements on Theorem~\ref{upper_bound_finite_distance}}
\label{remark_non_asymptotic_regularity}

We mention in the paper that it is possible to strengthen Theorem~\ref{upper_bound_finite_distance} at the price of more regularity for $K_U$.
Indeed, assume $K_U$ is three times differentiable and invoke a
second-order limited expansion at $(\UU_i,\UU_j)$ for all the maps $(\uu,\vv)\mapsto  K_U( \uu,\vv)-2K_U(\uu,\UU_j)  +K_U(\UU_i,\UU_j)$, $i,j\in \{1,\ldots,n\}$.
With the same reasoning as in the proof above, this yields
\begin{eqnarray*}
\lefteqn{ \Db^2 (\hat\Pb_n,\Pb_n) =
\frac{1}{n^2} \sum_{i,j=1}^n \Big\{
(\hat \UU_i - \UU_i)^\top \partial^2_{1,2} K_U( \UU^*_i,\UU^*_j) ( \hat\UU_j - \UU_j) }\\
&+&
(\hat \UU_i - \UU_i)^\top \big\{\partial^2_{1,1} K_U( \UU^*_i,\UU^*_j) - \partial^2_{1,1} K_U( \UU^*_i,\UU_j) \big\}(\hat \UU_i - \UU_i) \\
&=&
\frac{1}{n^2} \sum_{i,j=1}^n
(\hat \UU_i - \UU_i)^\top \partial^2_{1,2} K_U( \UU^*_i,\UU^*_j) ( \hat\UU_j - \UU_j)  \\
&+& \frac{1}{n^2} \sum_{i,j=1}^n
 \partial^3_{1,1,2} K_U( \UU^*_i,\tilde\UU_j)\cdot (\hat \UU_i - \UU_i)^{(2)}\cdot (\UU^*_j - \UU_j),
\end{eqnarray*}
since $ \partial^2_{1,1} K_U( \uu,\vv) = \partial^2_{2,2} K_U(\vv,\uu)$, with obvious notations for differentials. Then,
$$ \Db^2 (\hat\Pb_n,\Pb_n)\leq d^2 \|d^{(2)} K_U\|_\infty \sup_{i=1,\ldots,n} \sup_{k=1,\ldots,d} |\hat U_{ik} - U_{ik}|^2
+  d^3 \|d^{(3)} K_U\|_\infty \sup_{i=1,\ldots,n} \sup_{k=1,\ldots,d} |\hat U_{ik} - U_{ik}|^3 .$$
As above, we get with probability larger than $1-\delta-\nu$,
\begin{eqnarray*}
\lefteqn{ \Db(\Pb_{\hat\theta_n},\Pb_0) \leq \inf_{\theta\in\Theta} \Db(\Pb_{\theta},\Pb_0) +
  \Big\{ \frac{8}{n}\sup_{\uu\in [0,1]^d} K_U(\uu,\uu)  \Big\}^{1/2} \Big\{ 1+ \big(- \ln \delta \big)^{1/2} \Big\} }\\
  &+&
\bigg\{ \frac{d^2}{n} \|d^{(2)} K_U\|_\infty \ln\big(\frac{2d}{\nu}\big)\bigg\}^{1/2} + \Bigg[
\frac{d^3}{\sqrt{2}n^{3/2}} \|d^{(3)} K_U\|_\infty \bigg\{\ln\Big( \frac{2d}{\nu}\Big) \bigg\}^{3/2}\Bigg]^{1/2}.
\end{eqnarray*}

\section{Proof of Proposition~\ref{prop_TAYLOR}}
\label{proof_TAYLOR}

By a limited expansion, there exists $\theta^*$, $\|\theta^*-\theta_0\|\leq \|\theta-\theta_0\|$, such that
  \begin{align*}
\Db^2(\Pb_{\theta},\Pb_{\theta_0}) &
\geq (\theta-\theta_0)^\top \nabla_{\theta=\theta_0}  \Db^2(\Pb_{\theta},\Pb_{\theta_0}) + \frac{1}{4}(\theta-\theta_0)^\top \nabla^2_{\theta=\theta^*,\theta=\theta^*} \Db^2(\Pb_{\theta},\Pb_{\theta_0}) (\theta-\theta_0)
\\
& \geq 0 + \frac{\lambda_{\min}(\theta_0)}{4} \|\theta-\theta_0\|^2,
 \end{align*}
for any $\theta$ such that $\|\theta-\theta_0\|<r$,
  In other words, for such $\theta$, we have
 \begin{equation}
 \label{equa:compa:norm}
 \|\theta-\theta_0\| \leq \frac{2}{\sqrt{\lambda_{\min}(\theta_0)}} \Db(\Pb_{\theta},\Pb_{\theta_0}).
 \end{equation}
 The conditions on $\varepsilon$ and $n$ ensure that, on an event with probability $1-\nu-\delta$, $ \Db(\Pb_{\hat\theta_n},\Pb_{\theta_0}) < \alpha $,
 that is, $\|\hat\theta_n - \theta_0\|<r$. Thus, combining~\eqref{equa:compa:norm} with Theorem~\ref{upper_bound_finite_distance}, we get the result.

\section{Proof of Proposition~\ref{AN_Gaussian_cop}}
\label{proof_AN_Gaussian_cop}

For every $\theta$ in a sufficiently small open neighborhood of $\theta_0\in (-1,1)$, copula densities exist and we have
\begin{equation}
\ell(\ww;\theta) = \int K(\uu,\vv) c_\theta(\uu) c_\theta(\vv)\, d\uu \, d\vv - 2 \int K(\uu,\ww) c_\theta(\uu) \,d\uu .
\label{IPP_Gaussian_cop}
\end{equation}

Let us check that all conditions \ref{cond_consistency1}-\ref{cond_normality_4_bis} are satisfied in this case, to apply Theorem~\ref{Th_AN}.
\begin{itemize}

    \item Condition~\ref{cond_consistency1}: obviously, choose a compact set $\Theta=[-1+\varepsilon,1-\varepsilon]^2$ that includes $\theta_0$, for some sufficiently small $\varepsilon>0$. Use the identity~(\ref{IPP_Gaussian_cop}) and the dominated convergence theorem to prove that the map $\theta\mapsto L_0(\theta)$ is continuous on $\Theta$.
        Moreover, $L_0(\cdot)$ is uniquely minimized at $\theta_0$. Indeed, $L_0(\theta)$ is equal to the MMD distance between $C_\theta$ and $C_{\theta_0}$ (up to a constant), which is minimized at $\theta_0$ and nowhere else due to the identifiability of the Gaussian family and knowing that our kernel is characteristic.
    \item Condition~\ref{cond_consistency2}: for any $\theta \in \Theta$,
    $$ |  \ell(\ww;\theta) | \leq
    \int |K|(\uu,\vv) c_\theta(\uu) c_\theta(\vv)\, d\uu \, d\vv + 2 \int |K|(\uu,\ww) c_\theta(\uu) \,d\uu \leq 3.$$
    Thus, the envelope function of the family of functions $\ww\mapsto \ell(\ww,\theta)$ is a constant and is then integrable.
    By the dominated convergence theorem, $\theta\mapsto \ell(\ww,\theta)$ is continuous on $\Theta$.
    \item Condition \ref{cond_normality_0} is obviously satisfied with our choice $-1<\theta_0<1$.
    \item Condition \ref{cond_normality_1} and \ref{cond_normality_2} are satisfied. Indeed, we can calculate the derivatives of
    $\theta \mapsto \ell(\ww;\theta)$ by differentiating Gaussian copula densities
    inside the integral sign. This does not affect $K$. Since the latter kernel is bounded by one, the dominated convergence applies, even uniformly with respect to $\theta\in \Theta$ and $\ww\in [0,1]^d$.
    \item To get Condition \ref{cond_normality_3}, note that $B=\mathbb{E}[\nabla_{\theta,\theta}^2 \ell(\UU;\theta_0)] < +\infty$ with the same arguments as before.
    The calculation of $B$ is of interest, because it would yield an analytic form for the asymptotic variance of $\hat\theta_n$.
    As noted before, $B$ can be deduced from the map $\theta\mapsto \mathbb{E}[\ell(\UU;\theta)]$, after calculating the second derivative of the latter function,
    evaluated at $\theta=\theta_0$. This can be done when $K=K_U$, using the formulas and notations of Section~\ref{calc_criterion_Gaussian_cop}.
    Since
    $$ \mathbb{E}[\ell(\UU;\theta)] = \Ic(\theta,\theta)-2 \Ic(\theta,\theta_0)=I(\theta)- I\big( (\theta+\theta_0)/2 \big),$$
    with $ I(\theta)= \gamma^2\{ (2+\gamma^2/2)^2 - 4 \theta^2\}^{-1/2}/2$, we deduce $B= 3 I''(\theta_0)/4$. Simple calculations yield
$$ I'(\theta)= \frac{2\theta \gamma^2}{\{ (2+\gamma^2/2)^2 - 4 \theta^2\}^{3/2}}\;\;\text{and}\;\;
 I''(\theta)= \frac{2 \gamma^2 \big\{ (2+\gamma^2/2)^2 + 8\theta^2 \big\}}{\{ (2+\gamma^2/2)^2 - 4 \theta^2\}^{5/2}},$$
 that is strictly positive.
 \mds
 When $K=K_G$, no closed form formula for $B$ is available.

    \item Condition \ref{cond_normality_plus} is obviously satisfied (first-order conditions).
    \item Condition \ref{cond_normality_5_bis}:
     to prove that the gradient of the loss $\nabla_\theta\ell(\cdot;\theta_0)$ is of bounded variation, it is sufficient to show that the mixed partial derivative $\ww \mapsto \nabla^3_{\theta,1,2}\ell(\ww;\theta_0)$ is continuous on $[0,1]^2$, and also that the functions $w_1 \mapsto \nabla^2_{\theta,w_1}\ell(w_1,1;\theta_0)$ and $w_2 \mapsto \nabla^2_{\theta,w_2}\ell(1,w_2;\theta_0)$ are continuous on $[0,1]$ (\cite{Radoluvic}, p.3351).
     When $K=K_G$, there is no hurdle by dominated convergence. When $K=K_U$, this is guaranteed by the same argument when $\gamma^2<2$. Indeed, the latter condition
     imposes the nullity of the latter derivatives when one of their arguments is zero or one.

\item Condition \ref{cond_normality_4_bis} is satisfied for the Gaussian copula when $|\theta_0|<1$: see Example 5.1 in~\cite{Segers2012}.
\end{itemize}

\section{Proof of Corollary~\ref{AN_Gaussian_cop_cor}}
\label{AN_Gaussian_cop_cor_PROOF}

The arguments are exactly the same as for Proposition~\ref{AN_Gaussian_cop}, replacing $\theta_0$ by $\theta_0^*$ if necessary.
We need care only when expectations under the true DGP are required, i.e. in
Condition~\ref{cond_normality_2} mainly, the assumptions~\ref{cond_normality_3} and~\ref{cond_normality_plus} being assumed.
Since $C_0(d\uu)$ is a probability measure,
$\int \sup_{\theta\in N(\theta_0^*)} \big\|\nabla_{\theta,\theta}^2\ell (\uu;\theta) \big\| \, C_0(d\uu) < +\infty,$ for some neighborhood $N(\theta_0^*)$ of $\theta_0^*$, yielding Condition~\ref{cond_normality_2}.

\section{Proof of Proposition~\ref{AN_MO}}
\label{proof_AN_MO}
To apply Theorem~\ref{Th_AN}, it is sufficient to check that the conditions~\ref{cond_consistency1}-\ref{cond_normality_5_bis} and~\ref{cond_normality_4_ter}-\ref{cond_moment_dell} are satisfied.
To calculate $\ell (\cdot;\theta)$ and $L_0 (\theta)$, we rely on the formulas~(\ref{psi_UV}),~(\ref{def_I1}) and~(\ref{def_I2}).
Note that we will restrict ourselves to parameters $\alpha$ and $\beta$ into $[\epsilon,1-\epsilon]$.
Therefore,
$$ \bar u^*=\max_{(\alpha,\beta)\in \Theta} \bar u_{\alpha,\beta} < 1, \;\text{and}\;
 \bar u_*=\min_{(\alpha,\beta)\in \Theta} \bar u_{\alpha,\beta} >0.$$
It can be checked that the map $\bar u: (\alpha,\beta)\mapsto \bar u_{\alpha,\beta}=(\beta/\alpha)^{\beta/(\alpha-\beta)}$ from $\Theta$ to $\Rb$ is two times continuously differentiable. To this goal, it is necessary to extend the map $\bar u$ by continuity, setting $\bar u (\alpha,\alpha)=e^{-1}$,
$\partial_1 \bar u (\alpha,\alpha)=e^{-1}/(2\alpha)$,
$\partial_2 \bar u (\alpha,\alpha)=-e^{-1}/(2\alpha)$,
$\partial^2_{1,1} \bar u (\alpha,\alpha)=-5e^{-1}/(12\alpha^2)$, $\partial^2_{2,2} \bar u (\alpha,\alpha)=7e^{-1}/(12\alpha^2)$ and
$\partial^2_{1,2} \bar u (\alpha,\alpha)=-5e^{-1}/(12\alpha^2)$.

\mds
For any continuous and bounded map $\psi:[0,1]^2\mapsto \Rb$, we recall that
\begin{eqnarray*}
\lefteqn{ \Eb_\theta[\psi(U_1,U_2)] =
\int \psi (s,t)
\big\{ (1-\alpha)s^{-\alpha}\1(s^\alpha > t^\beta)+ (1-\beta)t^{-\beta}\1(s^\alpha < t^\beta) \big\} \,ds\,dt   }\\
& + &  \int_0^{\bar u_{\alpha,\beta}} \psi(u,u^{\alpha/\beta})\,\beta u^{1-\alpha}\, du +
\int_{\bar u_{\alpha,\beta}}^{1} \psi(u,u^{\alpha/\beta})\,\alpha u^{\alpha/\beta-\alpha}\, du , \hspace{4cm}
\end{eqnarray*}
that can be seen as the integral of a map $(s,t)\mapsto g_\theta(s,t)$ on $[0,1]^2$ with respect to the Lebesgue measure (single integrals are particular cases of double integrals!). Such maps are continuous a.e., and
$$ \sup_{\theta\in \Theta} |g_\theta|(s,t) \leq \|\psi \|_\infty \big( s^{\epsilon - 1} + t^{\epsilon-1} + 2  \big). $$
The function on the r.h.s. of the latter equation is integrable on $[0,1]^2$ with respect to the Lebesgue measure.
By dominated convergence, we deduce the map $\theta \mapsto \ell (\ww,\theta)=\Eb_\theta\big[K(\UU,\VV)\big] - 2 \Eb_\theta\big[K(\UU,\ww)\big]$
is continuous on $\Theta$ for every $\ww$, for any bounded kernel, in particular when $K\in \{K_U,K_G\}$. The same arguments apply for $\theta \mapsto L_0(\theta)$ when $\theta\in \Theta$.
We deduce that Conditions~\ref{cond_consistency1} and~\ref{cond_consistency2} are satisfied, and $\hat\theta_n$ is consistent.

\mds

Now assume $K=K_G$.
To check Condition~\ref{cond_normality_1}, we have to prove that the calculations of derivatives of $\ell(\ww,\theta)$ with respect to $\theta$ are permitted inside our integral signs.
Such integrands are indeed two times continuously differentiable with respect to $\theta\in \Theta$ for almost all their other arguments into the interior of their domains. Moreover, they are upper bounded by some integrable envelope functions. Then, the dominated convergence theorem applies.
Nonetheless, since these integrands are often integrals themselves, it may be necessary to rely on the dominated convergence theorem again to state continuity.
The calculations of such derivatives induce many terms, but the same technique applies to all.
We will illustrate the arguments on some of them.

\mds

For instance, the $\theta$-derivatives of $\ell(\ww,\theta)$ involves the derivatives of
$$ \theta \mapsto
I_1(\theta)= \int u_1^{-\alpha} v_1^{-\alpha} \Big\{ \int  K_G(\uu,\vv)  \1(u_1^{\alpha/\beta} > u_2,
v_1^{\alpha/\beta} > v_2) \, du_2\, dv_2 \Big\} \, du_1\, dv_1 .$$
The successive partial derivatives of $I_1$ with respect to $\alpha$ and/or $\beta$ can be obtained by derivation under the integral sign by applying the dominated
convergence theorem. Indeed, for any bounded map $H$ and any couple of nonnegative integers $(a,b)$,
$$ \sup_{(\alpha,\beta)\in \Theta} (u_1 v_1)^{-\alpha} |\ln u_1|^{a} |\ln v_1|^{b}
H(u^{\alpha/\beta}_1,v_1^{\alpha/\beta}) \leq \| H \|_\infty (u_1 v_1)^{-\epsilon} |\ln u_1|^{a} |\ln v_1|^{b}, $$
that is integrable. Here, the latter bounded map $H$ involves $K_G$, its derivatives and some nonnegative powers of its arguments.

\mds

Another term of $\ell(\ww;\theta)$ is
$$ \theta \mapsto I_2(\theta)= \int u^{1-\alpha} v_1^{-\alpha} K_G(u,u^{\alpha/\beta},\vv)  \1(u< \bar u_{\alpha,\beta},
v_2 < v_1^{\alpha/\beta}) \, du\, d\vv,$$
that can be managed similarly.

\mds

The last family of terms we have to manage are
\begin{eqnarray*}
\lefteqn{\theta \mapsto I_3(\ww,\theta)=
\int K_G (s,t,\ww)
\big\{ (1-\alpha)s^{-\alpha}\1(s^\alpha > t^\beta)+ (1-\beta)t^{-\beta}\1(s^\alpha < t^\beta) \big\} \,ds\,dt   }\\
& + &  \int_0^{\bar u_{\alpha,\beta}} K_G(u,u^{\alpha/\beta},\ww)\,\beta u^{1-\alpha}\, du +
\int_{\bar u_{\alpha,\beta}}^{1} K_G(u,u^{\alpha/\beta},\ww)\,\alpha u^{\alpha/\beta-\alpha}\, du , \hspace{4cm}
\end{eqnarray*}
for some $\ww \in [0,1]^d$. They do not induce any additional difficulty.
Then, Condition~\ref{cond_normality_1} is satisfied.
This is still the case for Condition~\ref{cond_normality_2} because the upper bounds of $I_3(\ww,\theta)$ and its partial derivatives with respect to $\theta$ are uniform
with respect to $(\alpha,\beta)\in \Theta$ and $\ww\in [0,1]^d$. The latter argument is key to justify Corollary~\ref{AN_MO_cop_cor}.

\mds

Condition~\ref{cond_normality_3} can be obtained by the same type of reasonings.
Nonetheless, we do not exclude that $B$ could be not invertible for particularly unhappy choices of $(\alpha,\beta,\gamma)$.
Since the latter set of parameters is the roots of some analytic expression $\Hc(\alpha,\beta,\gamma)=0$, its Lebesgue measure is zero most often.
Due to the regularity of $L_0$ and the correct model specification, Condition~\ref{cond_normality_plus} is fulfilled.

\mds

Again, Condition~\ref{cond_normality_5_bis} is satisfied still because
$\ww\mapsto \nabla^3_{\theta,w_1,w_2} \ell (\ww,\theta_0)$ is continuous on $[0,1]^2$ with respect to the Lebesgue measure and
$w_1 \mapsto \nabla^2_{\theta,w_1}\ell(w_1,1;\theta_0)$ and $w_2 \mapsto \nabla^2_{\theta,w_2}\ell(1,w_2;\theta_0)$ are continuous on $[0,1]$ with respect to the Lebesgue measure, by dominated convergence.
The same arguments apply to check Condition~\ref{cond_moment_dell}, by choosing the powers $q_I$ sufficiently close to one.

\mds

Finally, Condition~\ref{cond_normality_4_ter} is obviously satisfied because the curve $\Cf$ has Lebesgue measure zero on the plane.

\mds

\section{MMD criterion for a bivariate Gaussian copula model}
\label{calc_criterion_Gaussian_cop}

Here, we explicitly write our MMD criterion in the case of bivariate Gaussian copulas.
Recall that the density of a Gaussian copula in dimension two is
$$ c_\theta(u_1,u_2)= \frac{1}{2\pi\sqrt{1-\theta^2}\phi(x_1) \phi(x_2)} \exp\Big\{-\frac{1}{2(1-\theta^2)} \big( x_1^2+x_2^2 - 2\theta x_1 x_2  \big)   \Big\},$$
by setting $ x_k=\Phi^{-1}(u_k)$, $k=1,2$.
Define $\xx=(x_1,x_2)$. Similarly, $ y_k=\Phi^{-1}(v_k)$, $k=1,2$ and $\yy=(y_1,y_2)$.
For obtaining closed form formulas, it is necessary to select an adapted kernel.
Here, we use the Gaussian-type kernel~(\ref{def_K_U}), with $h=\Phi^{-1}$.

\mds

Now, let us analytically specify the criterion in~(\ref{RKHS_criterion}). First, let us calculate
    $$ \Ic(\theta_1,\theta_2)= \int K_U(\uu,\vv)\, C_{\theta_1}(d\uu) \, C_{\theta_2}(d\vv),\; (\theta_1,\theta_2)\in (-1,1)^2.$$
    By a change of variable, note that
    $$ \Ic(\theta_1,\theta_2)= \Eb\Big[ \exp\big\{-\frac{(X_1-Y_1)^2+ (X_2- Y_2)^2}{\gamma^2}\big\}\Big],$$
    for a Gaussian centered random vector $(X_1,X_2,Y_1,Y_2)$ whose $4\times 4$ covariance matrix is block-diagonal. Its first (resp. second) $2\times 2$ block is a correlation matrix with an extra-diagonal coefficient $\theta_1$ (resp. $\theta_2$). Therefore, the bivariate random vector $(Z_1,Z_2)=(X_1-Y_1,X_2-Y_2)/\sqrt{2}$ is centered Gaussian and its covariance matrix is a correlation matrix with an extra-diagonal coefficient $s=(\theta_1+\theta_2)/2$. Since the conditional law of $Z_1$ given $Z_2=z_2$ is $\Nc(s z_2,1-s^2)$,
    we can easily calculate $\psi(z)=\Eb\big[  \exp \{ -Z_1^2 /(\gamma^2/2) \} | Z_2=z \big]$.
    Indeed, setting $\tau^2 = \{ 1/\gamma^2+1/(1-s^2)\}^{-1}$, we have
    \begin{eqnarray*}
    \lefteqn{ \psi(z)=
    \int \exp\big(-\frac{t^2}{\gamma^2/2}\big) \exp\Big\{ -\frac{(t-s z)^2}{2(1-s^2)}\Big\} \frac{dt}{ \sqrt{2\pi} \sqrt{1-s^2}}     }\\
    &=& \int \exp\big(-\frac{t^2}{2\tau^2} + \frac{s t z }{1-s^2}\big) \frac{dt}{ \sqrt{2\pi} \sqrt{1-s^2}} \exp\big\{-\frac{s^2 z^2}{2(1-s^2)}    \big\}  \\
    &=& \
    \frac{\gamma/\sqrt{2}}{\sqrt{2(1-s^2)+\gamma^2/2}}   \exp\Big\{-\frac{s^2  z^2 }{2(1-s^2)+\gamma^2/2}\Big\}.
    \end{eqnarray*}
   We deduce
    \begin{eqnarray*}
    \lefteqn{ \Ic(\theta_1,\theta_2) = \Eb\Big[  \exp\big( -\frac{Z_1^2 + Z_2^2}{\gamma^2/2}   \big)  \Big]
    = \Eb_{Z_2}\Big[  \exp\big( -\frac{Z_2^2}{\gamma^2/2}   \big) \Eb\big[  \exp\big( -\frac{Z_1^2 }{\gamma^2/2}\big) | Z_2 \big]   \Big] }\\
    &=& \int \exp\big( - \frac{t^2}{\gamma^2/2}\big) \psi(t) \phi(t) \, dt
    = \gamma^2/2 \big\{ (2 + \gamma^2/2)^2 - 4s^2 \big\}^{-1/2}=: I(s).
    \end{eqnarray*}
Moreover, the other integrals in~(\ref{RKHS_criterion}) are as
$$ \int K_U(\uu,\hat \UU_i) c_\theta (\uu) \, d\uu=
\Eb \Big[ \exp\big\{- \frac{(X_1- \Phi^{-1}(\hat U_{i,1}))^2+(X_2
- \Phi^{-1}(\hat U_{i,2}))^2}{\gamma^2}     \big\}   \Big],$$
for some standardized bivariate Gaussian random vector $(X_1,X_2)$, $\Eb [X_1 X_2]=\theta$.
For any real numbers $(a,b)$, standard arguments yield
\begin{eqnarray*}
\lefteqn{\Jc (\theta,a,b)= \Eb \Big[ \exp\big\{- \frac{(X_1- a)^2+(X_2- b)^2}{\gamma^2} \big\} \Big]  }\\
&=& \Eb \Big[  \exp\big\{- \frac{(X_2- b)^2}{\gamma^2} \big\} \Eb\big[ \exp\big(- \frac{(X_1- a)^2}{\gamma^2} \big) | X_2\big] \Big]   \\
&=& \frac{\gamma/\sqrt{2}}{\sqrt{1+\gamma^2/2-\theta^2}}\Eb \Big[  \exp\big\{- \frac{(X_2- b)^2}{\gamma^2} \big\}
\exp\big\{ - \frac{(\theta X_2- a)^2}{2(1+ \gamma^2/2 - \theta^2)}    \big\} \Big]   \\
&=& \frac{\gamma/\sqrt{2}}{\sqrt{1+\gamma^2/2-\theta^2}}\int
\exp\big\{- \frac{x^2}{2g^2}+\frac{\lambda x}{g^2} -\frac{b^2}{\gamma^2} -\frac{a^2}{2(1+\gamma^2/2-\theta^2)}
 \big\}\,\frac{dx}{\sqrt{2\pi}}    \\
&=& \frac{g\gamma/\sqrt{2}}{\sqrt{1+\gamma^2/2-\theta^2}} \exp\Big\{ \frac{\lambda^2}{2g^2} - \frac{b^2}{\gamma^2} - \frac{a^2}{2(1+\gamma^2/2-\theta^2)}   \Big\},
\end{eqnarray*}
by setting
$$ \frac{1}{g^2} =\frac{1}{\gamma^2/2} + \frac{\theta^2}{1+\gamma^2/2-\theta^2} +1,\;\frac{\lambda}{g^2}= \frac{b}{\gamma^2/2} + \frac{a\theta}{1+\gamma^2/2 -\theta^2} \cdot$$
Therefore, the estimated parameter of the bivariate Gaussian copula is
$$ \hat\theta_n=\arg\min_\theta \Ic(\theta,\theta) - 2n^{-1}\sum_{i=1}^n \Jc\big(\theta,\Phi^{-1}(\hat U_{i,1}),\Phi^{-1}(\hat U_{i,2})\big).$$
Note that generalizations of the latter calculations in larger dimensions would be quite cumbersome.
Finally, let us notice that
\begin{align*}
 \mathbb{D}^{2}(\Pb_{\theta_1},\Pb_{\theta_2})
 & = \mathcal{I}(\theta_1,\theta_1) + \mathcal{I}(\theta_2,\theta_2) - 2 \mathcal{I}(\theta_1,\theta_2)\\
 & = f(\theta_1)+f(\theta_2)-2f\left(\frac{\theta_1+\theta_2}{2}\right),
\end{align*}
where
$$ f(x) =  \frac{\gamma^2/2}{\sqrt{(2+\gamma^2/2)^2 - 4 x^2  }} \cdot$$
As, for any $x\in(-1,1)$,
\begin{align*}
f''(x) &
= \frac{3x^2 \gamma^2/2}{\left\{(2+\gamma^2/2)^2 - 4 x^2  \right\}^{5/2}} + \frac{\gamma^2/2}{\left\{(2+\gamma^2/2)^2 - 4 x^2  \right\}^{3/2}} \\
& \geq \frac{\gamma^2/2}{\left(2+\gamma^2/2 \right)^3} =: \alpha(\gamma),
\end{align*}
we obtain that $f$ is $\alpha(\gamma)$-strongly convex. This leads to
$$
f\left(\frac{\theta_1+\theta_2}{2}\right) \leq \frac{f(\theta_1) + f(\theta_2)}{2} - \frac{\alpha(\gamma)}{8}(\theta_1-\theta_2)^2,$$
that implies
$$
(\theta_1-\theta_2)^2 \leq \frac{4}{\alpha(\gamma)} \mathbb{D}^{2}(\Pb_{\theta_1},\Pb_{\theta_2}).
$$
Therefore, we have obtained
$|\theta_1-\theta_2| \leq 2 \mathbb{D}(\Pb_{\theta_1},\Pb_{\theta_2})/ \sqrt{\alpha(\gamma)}$,
which proves the claim in Example~\ref{exm:Gaussian:robustness}, setting
$$ c(\gamma) = \frac{2}{\sqrt{\alpha(\gamma)}} = \frac{\left(4+\gamma^2\right)^{3/2}}{\gamma} \cdot $$

\section{Computational cost}
\label{comp_cost}

Let us provide some comments on the computational cost of our algorithm.
If the gradient descent involves $T$ steps, its total numerical cost is obviously $T$ times $Cost1Step$,
the cost of one gradient step. The latter one can be decomposed as $$Cost1Step= O\big(n \times CostSampling(d,p) + n^2 (p+d) + n \times CostGrad(d,p)\big)$$ where $CostSampling(d,p)$ is the cost of sampling one $d$-dimensional observation from the copula $C_{\theta}$ for a given value of $\theta$ and $CostGrad(d,p)$ is the cost of computing the gradient $d\ln c_\theta(\uu) / d\theta$ for one value of $\uu$ and one value of $\theta$.
We must have $CostSampling(d,p) \geq \max(d,p)$ and $CostGrad(d,p) \geq \max(d,p)$ but explicit expressions of these costs will depend on the form of the parametric family of copulas that is used. Note that in very high dimension (large $d$ and/or $p$), our procedure becomes slow. Under such circumstances, it would be necessary to rely on alternative optimization criteria, as in composite likelihood techniques or pair-copula
 constructions. Both methodologies estimate the copulas associated to subvectors of $\XX$ (bivariate ones for vines), a strategy that may be sufficient to identify the true parameter and is numerically relevant.
 Table~\ref{tab:computation_time} displays the computation time of our algorithms for different parametric families. Note that this is much slower than the \texttt{BiCopEst} method in the \texttt{VineCopula} package which typical running time for the CML estimator with $n=1000$ is 0.01 second. This difference may be due to the programming language since \texttt{VineCopula}'s underlying code is written in C.
All these computations were done on a Windows 10 laptop, with a processor Intel Core i7-3630 2.40GHz.

\begin{table}[t]
    \centering 
    \begin{tabular}{lll}
        \textbf{Family} & \textbf{Mean computation time (s)} & \textbf{Sd computation time (s)} \\
        \hline
Gaussian &	4.80 &	0.227	\\	
Clayton &	5.22 &	0.177		\\		
Gumbel &	7.50 &	0.199		\\		
Frank &	72.16 &	1.397	\\	
Marshall-Olkin &	4.26 & 	0.217		\\	
    \end{tabular}
       \begin{tabular}{lll} 
        \textbf{Family} & \textbf{Mean computation time (ms)} & \textbf{Sd computation time (ms)} \\
                \hline
Gaussian & 5.502 & 0.73 \\
Clayton  & 6.059 & 0.40 \\
Gumbel   & 11.56 & 0.23 \\
Frank    & 7.704 & 0.86 \\
Marshall-Olkin & (not implemented) & (not implemented)
\end{tabular}
    
    \caption{Computation time of the MMD estimator (first panel) and of \texttt{BiCopEst} (second panel) for different families (with $n=1000$)}
    \label{tab:computation_time}
\end{table}

\section{Supplementary experiments: confidence intervals by bootstrap and subsampling}
\label{bootstrap}

In this section, we compare the empirical properties of two confidence intervals, one based on bootstrap and the other one based on subsampling.
We find that the bootstrap-based confidence intervals (when the subsampling size is the same as the sample size) is too liberal and does not attain its nominal 95\% coverage.
At the opposite, the subsampling-based confidence intervals (when the subsampling size is smaller than the sample size) are quite conservative with 100\% coverage in the Gaussian case: see Figure~\ref{fig:length_CI} and Figure~\ref{fig:coverage_CI}.

\begin{figure}[htb]
    \begin{minipage}[t]{0.49\textwidth}
        \begin{center}
 \includegraphics[width=\textwidth]{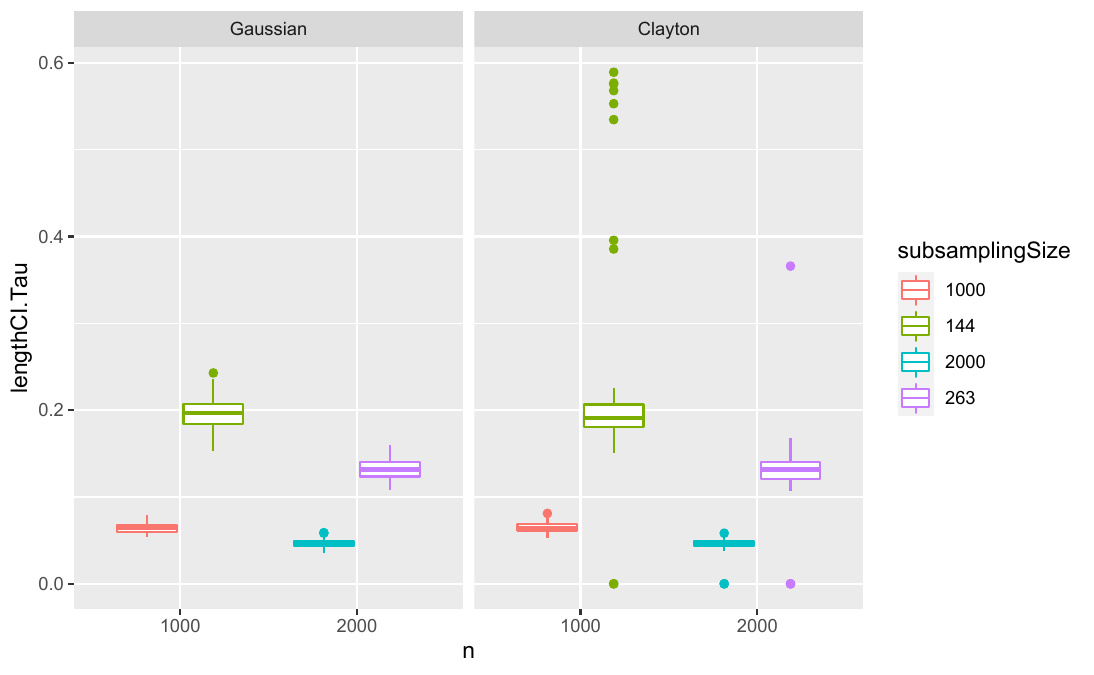}
 \caption{Boxplot of the length of the confidence intervals for two sample sizes.}
 \label{fig:length_CI}
 \end{center}
    \end{minipage}
    \hfill
    \begin{minipage}[t]{0.49\textwidth}
       \begin{center}
 \includegraphics[width=\textwidth]{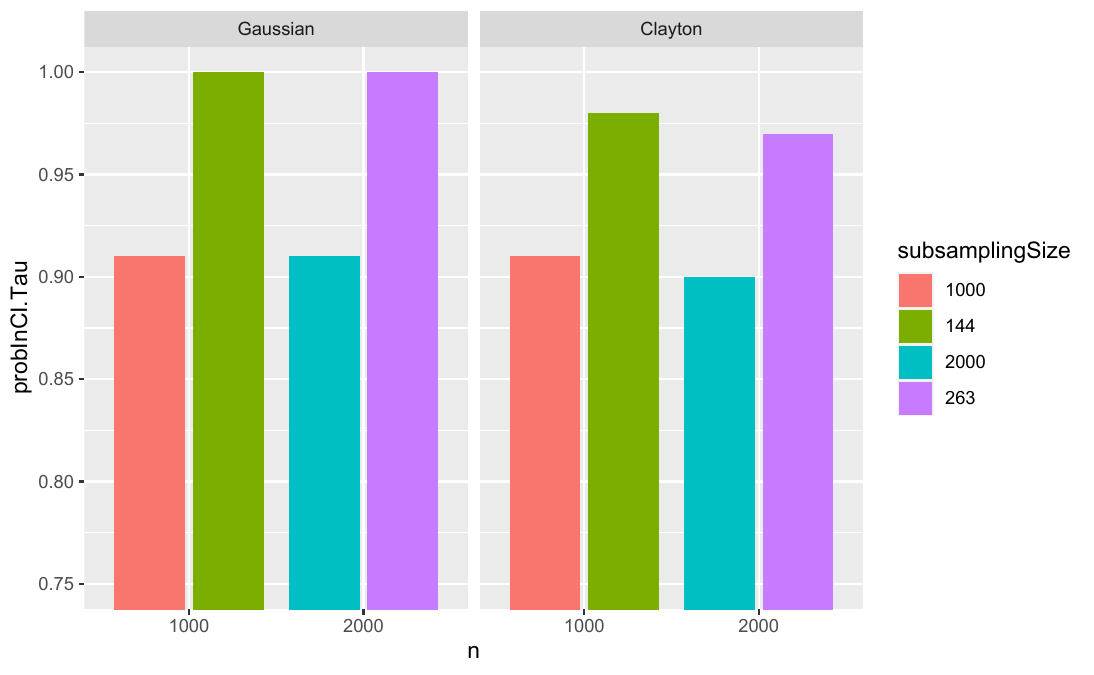}
 \caption{Barplot of the coverage percentage of the confidence intervals for two sample sizes.}
 \label{fig:coverage_CI}
 \end{center}
    \end{minipage}
\end{figure}

\section{Supplementary experiments: influence of $\gamma$}
\label{supp_gamma}

In these supplementary experiments, we evaluate the influence of the tuning parameter $\gamma$ of the MMD estimator, along with other properties of the estimation algorithm. In Figure~\ref{FigureMSE_gamma_init}, two methods of initialisation are compared (with a sample size $n=1000$). In the first method the empirical Kendall's tau is used as the starting point while in the second, a random starting point is sampled from the interval $[-0.95, 0.95]$.

In the non-contaminated case, the estimator initialised with the empirical Kendall's tau exhibit the best performances
even when $\gamma$ is very low, since the starting point is already very good.
On the contrary, in the contaminated case with 50 outliers out of 1000 data points, the estimator with such an initialisation
method is actually the worst since the empirical Kendall's tau may correspond in this case to a local minimum only,
instead of a global one. Note that these comparisons concerns the optimal $\gamma$ (i.e. in a minimax sense).
As soon as $\gamma$ is far from its optimal value, the initialisation with Kendall's tau is the best one, even in the
contaminated case. Therefore, we have decided to choose the optimal value of $\gamma$ and
the random initialisation by default, since the goal of this article is to construct a
robust estimator with the best performance.


\begin{figure}[htb]
 \begin{center}
 \includegraphics[width=15cm]{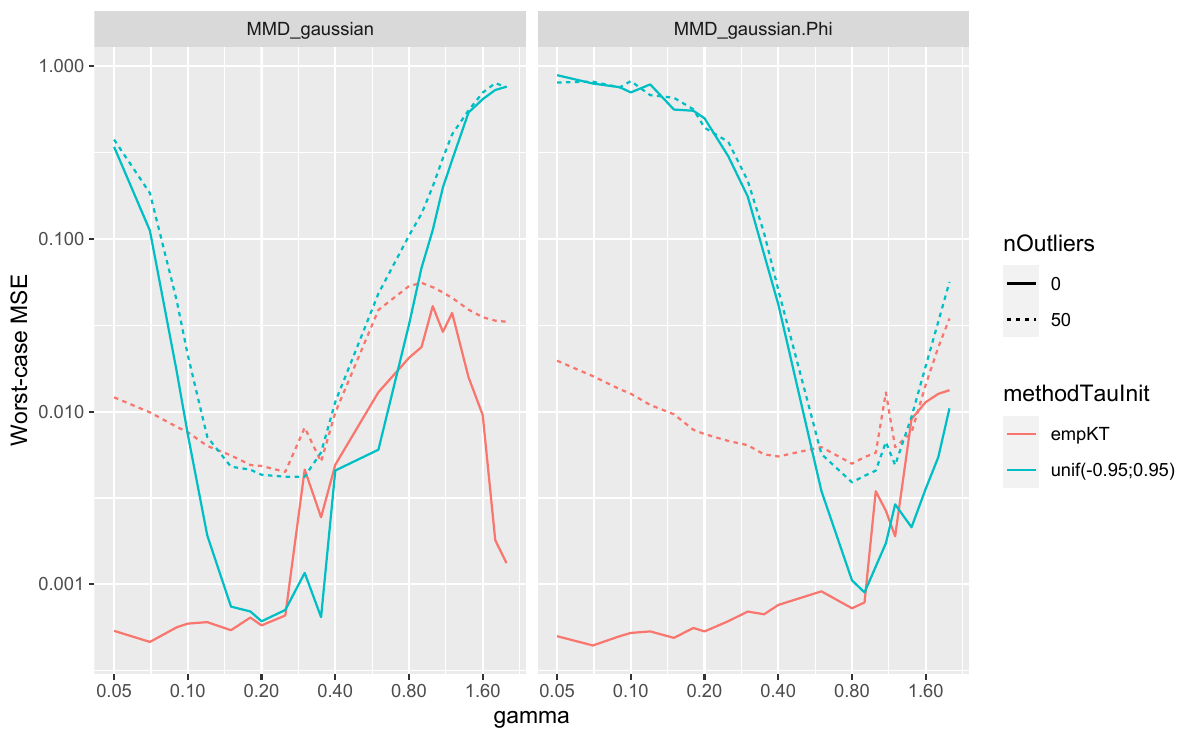}
 \caption{MSE of the MMD estimator with Gaussian kernel and $h(u)=u$, the MMD estimator with Gaussian kernel and $h(u)=\Phi^{-1}(u)$ as a function of the tuning parameter $\gamma$ for two different numbers of outliers and two different methods of initialisation: the initialisation using the empirical Kendall's tau (``empKT'') and the one using a random number from the interval $[-0.95, 0.95]$ (``unif(-0.95;0.95)'').}
 \label{FigureMSE_gamma_init}
 \end{center}
\end{figure}

\mds

In order to have a more precise understanding on the link between $\gamma$ and the MSE, two files containing 3-dimensional plots are given as Supplementary Material. The first one describes the influence of $\gamma$ in the regular parametric copula families (Gaussian, Clayton, Gumbel, Frank) on the MSE of the MMD with $K_U$ while the second concerns Marshall-Olkin copulas for the two kernels $K_U$ and $K_G$. For each estimator and each model, two plots are given on a logarithmic and on a linear scale, for easier comparison.


\end{document}